\newif\iffull
\newif\ifsubmission

\fulltrue

\documentclass[titlecompat,tight,letterpaper,twocolumn,10pt]{article}
\usepackage{usenix}

\PassOptionsToPackage{hyphens}{url}

\usepackage{xspace}
\usepackage{graphicx}
\usepackage{enumitem}
\usepackage{booktabs}
\usepackage{multirow}
\usepackage[binary-units=true]{siunitx}
\usepackage{amsfonts,amsthm,amsmath}
\usepackage{subcaption}
\usepackage{xcolor}
\usepackage{siunitx}
\usepackage{framed}
\usepackage{bm}
\usepackage[ruled,algosection,noend]{algorithm2e}
\usepackage{algpseudocode}
\usepackage[all]{xy}
\DeclareMathAlphabet{\mathsc}{OT1}{cmr}{m}{sc}
\usepackage{authblk}

\newtheorem{theorem}{Theorem}[section]
\newtheorem{definition}[theorem]{Definition}
\newtheorem{lemma}[theorem]{Lemma}
\newcommand\sarah[1]{\textcolor{red}{Sarah: #1}}
\newcommand\pavel[1]{\textcolor{blue}{Pavel: #1}}
\newcommand\cindy[1]{\textcolor{magenta}{Cindy: #1}}
\newcommand\mariana[1]{\textcolor{green}{Mariana: #1}}
\newcommand\gary[1]{\textcolor{teal}{Gary: #1}}
\newcommand\al[1]{\textcolor{orange}{Al: #1}}

\renewcommand\pavel[1]{}
\renewcommand\sarah[1]{}
\renewcommand\cindy[1]{}
\renewcommand\mariana[1]{}
\renewcommand\gary[1]{}
\renewcommand\al[1]{}

\usepackage{hyperref}

\newcommand*\dash{\ifvmode\quitvmode\else\unskip\kern.16667em\fi---%
\hskip.16667em\relax}

\newenvironment{sarahlist}{
\begin{description}[itemsep=2pt,leftmargin=0.4cm]
}{\end{description}}

\newcommand{\phead}[1]{\noindent{\bf\normalsize #1.}}

\newenvironment{gamespec}{
\begin{center}
\begin{tabular}{l}
}{\end{tabular}\end{center}}

\newcommand{\trm}[1]{\textrm{#1}}

\newcommand{\N}{\mathbb{N}}

\newcommand{\A}{\mathcal{A}}
\newcommand{\B}{\mathcal{B}}
\newcommand{\pr}{\textrm{Pr}}

\newcommand{\secp}{\lambda}
\newcommand{\usecp}{1^\lambda}

\newcommand{\keygen}{\mathtt{KeyGen}}

\newcommand{\verify}{\mathtt{Verify}}
\newcommand{\sign}{\mathtt{Sign}}
\newcommand{\prove}{\mathtt{Prove}}

\newcommand{\old}{\mathsf{old}}
\newcommand{\new}{\mathsf{new}}

\newcommand{\pk}{\mathsf{pk}}
\newcommand{\sk}{\mathsf{sk}}

\newcommand{\randpick}{\mathrel{\xleftarrow{\$}}}

\newcommand{\adv}{\mathbf{Adv}}
\newcommand{\game}[3]{\mathsf{G}_{#1}^{#2}(#3)}

\newcommand{\main}{\mathsc{main}}

\newcommand{\oracle}{\mathcal{O}}
\newcommand{\advgame}[3]{\adv^{#1}_{#2}(#3)}

\newcommand{\splitview}[2]{\game{#1}{\mathsf{SVA}}{#2}}
\newcommand{\advsplitview}[2]{\adv^{\mathsf{SVA}}_{#1}(#2)}
\newcommand{\oscillation}[2]{\game{#1}{\mathsf{Osc}}{#2}}
\newcommand{\advoscillation}[2]{\adv^{\mathsf{Osc}}_{#1}(#2)}

\newcommand{\ks}{\mathsf{k\mbox{-}snd}}

\newcommand{\advmemb}[2]{\advgame{\mathsf{memb}}{#1}{#2}}
\newcommand{\advappend}[2]{\advgame{\mathsf{append}}{#1}{#2}}
\newcommand{\advlookup}[2]{\advgame{\mathsf{lookup}}{#1}{#2}}
\newcommand{\advks}[2]{\advgame{\mathsf{k\mbox{-}snd}}{#1}{#2}}
\newcommand{\advcr}[2]{\advgame{\mathsf{cr}}{#1}{#2}}

\newcommand{\clist}{\mathsf{list}}

\newcommand{\proveincl}{\mathtt{ProveIncl}}
\newcommand{\proveappend}{\mathtt{ProveAppend}}
\newcommand{\verincl}{\mathtt{VerIncl}}
\newcommand{\verappend}{\mathtt{VerAppend}}

\newcommand{\append}{\mathtt{Append}}
\newcommand{\vercommit}{\mathtt{VerCom}}
\newcommand{\lookup}{\mathtt{Lookup}}
\newcommand{\verlookup}{\mathtt{VerLookup}}
\newcommand{\histlookup}{\mathtt{Hist}}
\newcommand{\verhistlookup}{\mathtt{VerHist}}
\newcommand{\audit}{\mathtt{Audit}}
\newcommand{\veraudit}{\mathtt{VerAudit}}
\newcommand{\version}{\mathtt{Vrsn}}

\newcommand{\snap}{\mathsf{chkpt}}
\newcommand{\snaps}{\mathsf{chkpts}}
\newcommand{\entry}{\mathsf{entry}}
\newcommand{\entries}{\mathsf{entries}}

\newcommand{\mmap}{\mathsf{map}}
\newcommand{\mlog}{\mathsf{log}}
\newcommand{\mtable}{\mathsf{table}}
\newcommand{\mrl}{\mathsf{MRL}}
\newcommand{\croot}{\mathsf{root}}
\newcommand{\roothash}{h}
\newcommand{\maproot}{\roothash_\mmap}

\newcommand{\leaflog}{\mathsf{leaflog}}

\newcommand{\dict}{\mathsf{D}}
\newcommand{\key}{\mathsf{key}}
\newcommand{\val}{\mathsf{val}}
\newcommand{\hist}{\mathsf{hist}}
\newcommand{\histrepr}{\mathsf{hist\_rep}}

\newcommand{\btwn}{\mathsf{btwn}}

\newcommand{\clientcheck}{\mathtt{CheckEntry}}
\newcommand{\clientupdate}{\mathtt{UpdateChkpt}}
\newcommand{\ocheck}{\oracle_{\mathsf{check}}}
\newcommand{\oupdate}{\oracle_{\mathsf{update}}}

\newcommand{\range}{\mathsf{R}}
\newcommand{\compact}[1]{\mathsf{cmpct}#1}
\newcommand{\hgt}{\mathsf{hgt}}
\newcommand{\node}{\mathsf{v}}
\newcommand{\parent}{\mathsf{prnt}}
\newcommand{\coverage}{\mathsf{cover}}

\newcommand{\compactrangealg}{\textsc{CompactRange}\xspace}

\newcommand{\mergealg}{\textsc{Merge}\xspace}

\newcommand{\rangerootalg}{\textsc{RangeToRoot}\xspace}
\newcommand{\mth}{\textsc{MTH}}

\newcommand{\numsub}[1]{\mathsf{num}_{#1}}
\newcommand{\numwitnesses}{N_W}

\newcommand{\numreqsnap}{\gamma}

\newcommand{\mogname}{\textrm{Mog}\xspace}

\begin{document}

\title{Think Global, Act Local: Gossip and Client Audits in Verifiable Data Structures}

\ifsubmission
\else{
\def\authspace{\hspace{10pt}}
\def\google{\raisebox{5pt}{\small $\diamond$}}
\def\unaffil{\raisebox{5pt}{\tiny $\triangle$}}
\author{Sarah Meiklejohn\google\thanks{Contact email: \url{meiklejohn@google.com}}, 
Pavel Kalinnikov\google, 
Cindy S. Lin\google, 
Martin Hutchinson\google, 
Gary Belvin\unaffil, 
Mariana Raykova\google, 
Al Cutter\google \\

\google Google LLC \authspace
\unaffil Unaffiliated \authspace \\
}
}\fi

\maketitle

\begin{abstract}

In recent years, there has been increasing recognition of the benefits of having services provide auditable logs of data, as demonstrated by the deployment of Certificate Transparency and the development of other transparency projects.  Most proposed systems, however, rely on a \emph{gossip} protocol by which users can be assured that they have the same view of the log, but the few gossip protocols that do exist today are not suited for near-term deployment.  
Furthermore, they assume the presence of global sets of auditors, who must be blindly trusted to correctly perform their roles, in order to achieve their stated transparency goals.
In this paper, we address both of these issues by proposing a gossip protocol and a \emph{verifiable registry}, \mogname, in which users can perform their own auditing themselves.  We prove the security of our protocols %
and demonstrate via experimental evaluations that they are performant in a variety of potential near-term deployments.  
\end{abstract}

\section{Introduction}

The introduction of end-to-end encrypted messaging by applications such as WhatsApp and Signal comes with the guarantee that conversations between users are kept private.  It is equally important, however, to guarantee that these conversations are being held with the intended recipient, meaning verifying that they agree on the same public keys to identify each other.  Messaging apps can fall victim to man-in-the-middle attacks otherwise, yet they currently prevent this with techniques such as manual inspection or QR code scanning that put the burden of verification on individual users~\cite{unicorn-chi}.

In recent years, various solutions have aimed to solve this problem by providing Key Transparency ~\cite{coniks,ethiks,SP:TomDev17,seemless}: the public keys used to identify users are stored with a server, and clients can perform \emph{lookups} to find the public keys for their contacts.  The set of keys belonging to a user may expand and change, however, reflecting their usage of multiple devices or replacement of existing ones.  To avoid having to trust the server to give out the right key at the right time, clients may perform periodic \emph{audits} of their own keys, which means looking at \emph{all} the public keys held for them by the server.  
Crucially, both of these operations should be performed in a \emph{verifiable} way, meaning it should be difficult for an untrusted server to give an invalid or incomplete response to a client.

Beyond KT, there are a growing number of projects with similar goals: Certificate Transparency (CT)~\cite{6962,6962bis,ESORICS:DGHS16,enhanced-ct,revocation-transparency,aki,arpki}, software transparency~\cite{CCS:FDPFSS14,chainiac,contour}, and more general transparency~\cite{SP:TomDev17,continusec}.  Looking further into the future, one could imagine bringing transparency to aspects of government; e.g., a land registry in which anyone could obtain verifiable results of a lookup (finding the current owner of a property) or an audit (finding the entire history of ownership).  
\mariana{Should we mention explicitly applications to PKI and then we can refer to the secure aggregation protocol https://eprint.iacr.org/2017/281.pdf (used in Federated Learning) which would obtain stronger (active) security assuming a PKI}

While many existing systems acknowledge that individual users must be responsible for monitoring their own entries (as, in KT for example, only they know if they have recently changed or added a key), they still assume that global \emph{auditors} act to verify the honest behavior of the server.  It is not possible to check if auditors are actually doing this, however, which means they must be trusted without the ability to verify.  If an application domain does not have a natural choice for such global auditors, or if a user is not willing to blindly trust a global auditor in this way, their only choice is to audit the entire data structure themselves, despite the fact that they might care only about their own entries.  This global audit is likely to be prohibitively expensive for an individual user.

More broadly, these systems all rely on the ability of users to be sure that they see the same data as others when accessing the server.  This is typically achieved via a \emph{gossip} protocol, which prevents servers from carrying out \emph{split-view attacks}.  Despite the development of various gossip protocols in recent years and the need for gossip to fulfill the promises of CT, there is ``next to no deployment in the wild'' of any gossip protocol~\cite{ct-pam}.  This is likely due to the fact that existing proposals rely on significant changes to the Internet infrastructure~\cite{ct-gossip,gossiping-in-ct,dahlberg-aggregation}, communication between many different participants~\cite{coniks}, or the usage of blockchains~\cite{SP:TomDev17,keybase-blockchain,contour}, all of which present obstacles in practical near-term deployments.

\vspace{2mm}
\phead{Our contributions}
In this paper, we present two distinct contributions.  First, in Section~\ref{sec:gossip} we propose a gossip protocol for verifiable logs that can be used to support a variety of different applications (CT, KT, etc.).  Our approach is most similar to the ``collective signing'' (CoSi) protocol due to Syta et al.~\cite{SP:STVWJG16}, in which a large set of \emph{witnesses} interact to produce a signature over some given message.  In our case, the message represents the latest \emph{checkpoint} of a verifiable log, which the witnesses should first ensure is valid (i.e., consistent with previous checkpoints).  Our goal is not to scale to a large number of witnesses; in fact, in scenarios like CT log servers, which are required to be organizationally independent,\footnote{\url{https://chromium.github.io/ct-policy/log_policy.html}} may even act as witnesses for each other~\cite{minimal-gossip}.  As compared to CoSi, this means we can eliminate the interaction between witnesses and have each one just sign the checkpoint, with the client then responsible for gathering these signatures and deciding whether or not some form of consensus has been reached. This has the added benefit that clients can have different policies around what constitutes consensus, allowing some clients to impose more stringent requirements or require signatures from a particular subset of witnesses.  We prove that this protocol prevents split-view attacks as long as some threshold number of witnesses are honest, and achieves liveness as long as witnesses satisfy some minimum uptime requirement.  

Next, in Section~\ref{sec:registry} we turn our attention to verifiable \emph{registries} of entries, which can consist of lists of public keys (in the case of KT), owners of a property (in the case of a land registry), or whatever is needed to support a given application.  We present \mogname, a verifiable registry that takes advantage of \emph{compact ranges} in Merkle trees, which we describe in Section~\ref{sec:compact} and may be of independent interest, to achieve efficiency.  To the best of our knowledge, \mogname is the first data structure to allow for \emph{personal} auditing, meaning a client can audit only its own part of the data structure in a way that is significantly more efficient than auditing the entire data structure.
\mogname thus achieves (provable) security without relying on any global auditors.

\section{Background and Definitions}\label{sec:defns}

\subsection{Preliminaries}\label{sec:notation}

For a finite set $S$, 
$|S|$ denotes its size and $x\randpick S$ denotes sampling a
member uniformly from $S$ and assigning it to $x$.  For an ordered list $\clist$ of objects, $\clist[i]$ denotes the $i$-th object; similarly, for an object $\mathsf{obj}$, $\mathsf{obj}[\mathsf{cmpnt}]$ denotes the subcomponent $\mathsf{cmpnt}$.  $\secp\in\N$
denotes the security parameter and $\usecp$ denotes its unary representation,
and $\varepsilon$ denotes the empty string.
PT stands for polynomial time.  By $y\gets A(x_1,\ldots,x_n)$ we denote running
algorithm $A$ on inputs $x_1,\ldots,x_n$ and assigning
its output to $y$, and by $y\randpick A(x_1,\ldots,x_n)$ we denote running
$A(x_1,\ldots,x_n;R)$ for a uniformly random tape $R$.
Adversaries are modeled as randomized algorithms.  %
We use code-based games in our security definitions~\cite{EC:BelRog06}.
A game $\game{\A}{\mathsf{sec}}{\secp}$, played with respect to a security
notion $\mathsf{sec}$ and adversary $\A$, has a $\main$ procedure whose
output is the output of the game.
$\pr[\game{\A}{\trm{sec}}{\secp}]$ denotes the probability that this
output is $1$.

\subsection{History trees}\label{sec:history-trees}

A Merkle tree~\cite{merkle} is a binary tree whose leaves represent some set of values.  More precisely, it is a set of nodes, where each node $\node$ is annotated with the following information: (1) a pointer $\node[\parent]$ to its parent; (2) a pointer $\node[\ell]$ to its left child; (3) a pointer $\node[r]$ to its right child; and (4) a hash label $\node[h]$.  We use $\hgt(\node)$ to denote the height of the subtree under $\node$ and use $N$ to denote the \emph{size} of the tree; i.e., the number of leaves it contains.  For the leaf at index $i$, the hash label is $H(\entry_i)$, where $\entry_i$ is a value and $H(\cdot)$ is a collision-resistant hash function.  For non-leaf nodes, the hash label is $H(\node[\ell][h] \| \node[r][h])$.  The \emph{root} $\croot$ of the tree is the node at height $\lceil\log(N)\rceil$.  Merkle trees enable efficient \emph{inclusion} proofs, which prove that an entry is in the tree, and \emph{consistency} proofs, which prove that a tree of size $N$ has been obtained from a tree of size $N-\ell$ only by appending entries to the tree (and in particular not deleting or modifying any of the existing entries).  These proofs are of size $O(\log N)$, and can be verified in time $O(\log N)$ against just the hash of the root(s).

A \emph{history} tree~\cite{usec:crowal09} (also recently called a Merkle Mountain Range~\cite{flyclient}) is a Merkle tree in which the left subtree is a perfect tree of size $2^i$ for $i = \lfloor\log(N-1)\rfloor$ and the right subtree is a history tree of size $N-2^i$.  
If the right subtree becomes ``full'' (meaning its size is also $2^i$), then it is incorporated into the left subtree (meaning $\croot$ becomes the root of the left subtree) and the right subtree is then empty and ready for new entries.  
This means the tree is not always balanced, but that append operations are efficient and the roots of perfect subtrees are ``frozen'' in place even as new values are appended.  
We exploit this property in our treatment of \emph{compact ranges} in Section~\ref{sec:compact}.

\subsection{Arguments of knowledge}\label{sec:snark}

A non-interactive argument of knowledge for a relation $R$ consists of two algorithms: $\pi\randpick\prove(R,x,w)$, which outputs a proof $\pi$ that $(x,w)\in R$; and $0/1\gets\verify(R,x,\pi)$, which outputs a bit indicating whether or not the proof $\pi$ is convincing that $x\in L_R$.  The argument is \emph{correct} if $\verify(R,x,\prove(R,x,w))=1$ for all $(x,w)\in R$ and satisfies \emph{knowledge soundness} if it is hard for a malicious prover to convince a verifier of an incorrect statement.  More formally, this says for all PT adversaries $\A$ and instances $x$ there exists a PT extractor $\chi$ such that the probability is negligible that $\A$ produces a proof $\pi$ such that (1) $\verify(R,x,\pi)=1$ but (2) $(x,w)\notin R$ for $w\gets\chi(x,\pi)$.  Additionally, the argument is \emph{succinct} (or is a \emph{SNARK}~\cite{Groth2010}) if it has a small constant proof size and verification runs in constant time.

\subsection{Gossip-based verifiable logs}\label{sec:gossip-defs}

When using a \emph{verifiable log}, a server maintains the log $\mlog$ and a client, who does not necessarily trust the server, maintains a succinct \emph{checkpoint} $\snap$ of the log.  
Each batch of updates to the log results in a new \emph{version}, which should be uniquely recoverable from a checkpoint using the function $\version(\snap)$.  
We consider five algorithms associated with a verifiable log.

\begin{sarahlist}

\item[$\snap,\mlog\randpick\append(\mlog,\entries)$] is used to append a new batch of entries to the log, which results in a new version and thus a new checkpoint.  We occasionally overload notation and write $\snap_\new\randpick\append(\snap,\entries)$ when we care only about how the checkpoint changes.

\item[$\pi\gets \proveincl(\mlog,\entry)$] is used to provide a proof $\pi$ of the inclusion of a given entry in the log.

\item[$0/1\gets\verincl(\snap,\entry,\pi)$] is used to check the proof $\pi$ against a checkpoint $\snap$. 

\item[$\pi\gets\proveappend(\mlog_\new,\snap,\snap_\new)$] is used to prove that a new checkpoint $\snap_\new$ has been obtained by only appending entries to the log associated with an old checkpoint $\snap$. 

\item[$0/1\gets\verappend(\snap,\snap_\new,\pi)$] is used to check the proof $\pi$ concerning the \emph{consistency} of $\snap_\new$ with $\snap$.

\end{sarahlist}

Additionally, we consider the following two interactive protocols run between a client in possession of a checkpoint $\snap$ and a server in possession of the full log $\mlog$.

\begin{sarahlist}
\item[$\clientupdate$] allows the client to update its locally stored checkpoint $\snap$.  

\item[$\clientcheck$] allows the client to check, against its current checkpoint, whether or not an entry is in the log.  This typically means the client provides $\entry$, the server produces $\pi\gets\proveincl(\mlog,\entry)$, and the client checks $\verincl(\snap,\entry,\pi)$.
\end{sarahlist}

\newcommand\formallogdefns{
We use the same definition of correctness as Tomescu et al.~\cite[Definition 3.1]{CCS:TBPPTD19}, which says that $\verincl(\snap,\entry,\pi)=1$ for $\pi\gets\proveincl(\mlog,\allowbreak\entry)$ and $\verappend(\snap,\allowbreak\snap_\new,\pi) =1$ for $\pi\gets\proveappend(\mlog_\new,\allowbreak\snap,\allowbreak\snap_\new)$ when $\snap,\mlog$ and $\snap_\new,\mlog_\new$ are produced by a series of $\append$ operations (with $\mlog_\new$ at a later version than $\mlog$) that include $\entry$.
In terms of efficiency, the goal of a verifiable log is to have the runtime of all verification algorithms be $O(\log N)$, and to have all proofs be of size $O(\log N)$ as well.  
In terms of security, we cannot use the definitions of Tomescu et al., which rely on the ability to prove non-inclusion (which cannot be done efficiently in a log), so instead follow Chase and Meiklejohn~\cite{CCS:ChaMei16} in formalizing the necessary security properties.  We first define an additional algorithm $0/1\gets\vercommit(\snap,\entries)$ that is used to check if a checkpoint $\snap$ commits to a list of entries $\entries$. (Unlike the other verification algorithms, we do not expect the client to run this; it is defined just for notational convenience.)

\begin{definition}\textbf{\emph{\cite{CCS:ChaMei16}}}\label{def:append-only-log}
Define $\advappend{\A}{\secp}$ as the probability of an adversary $\A$ outputting $(\snap_1,\snap_2,\entries,\pi)$ such that (1) $\vercommit(\snap_2,\entries)$, (2) $\verappend(\snap_1,\allowbreak\snap_2,\pi) = 1$, and (3) there does not exist a prefix of $\entries$ to which $\snap_1$ commits; i.e., there does not exist an index $j$ such that $\vercommit(\snap_1,\entries[1:j])$.  If for all PT adversaries $\advappend{\A}{\sec} < \nu(\secp)$ for some negligible function $\nu(\cdot)$, then the log satisfies \emph{append-only security}.
\end{definition}

\begin{definition}\textbf{\emph{\cite{CCS:ChaMei16}}}\label{def:membership-log}
Define $\advmemb{\A}{\secp}$ as the probability of an adversary $\A$ outputting $(\snap,\entry,\entries,\pi)$ such that (1) $\vercommit(\snap,\entries)$, (2) $\verincl(\snap,\entry,\pi) = 1$, but (3) $\entry\notin\entries$.  If for all PT adversaries $\advmemb{\A}{\sec} < \nu(\secp)$ for some negligible function $\nu(\cdot)$, then the log satisfies \emph{membership security}.
\end{definition}

\mariana{I think it will be nice to write the above two definitions in the style of the next definition with the security games, this is just about the style}

These definitions prevent the server from giving conflicting responses to a single client, but in a setting with multiple clients the server may give out different checkpoints to different clients.  We thus consider the potential for the server to carry out a split-view attack, in which these different checkpoints are in fact inconsistent; i.e., commit to logs with different contents.  We formalize this attack as follows.

\begin{definition}[Split-view attack]\label{def:split-view}
Define $\advsplitview{\A}{\secp}$ as the probability of an adversary $\A$ winning the following game:
\begin{gamespec}
\underline{$\main$ $\splitview{\A}{\secp}$}\\
$A,B\gets\emptyset$; $\snaps\gets\vec{\varepsilon}$\\
$\entries\randpick\A^{\oupdate,\ocheck}(\usecp)$\\
return $(\exists (\snap_0,\entry,\pi)\in A ~\land~\snap_1\in B ~|~$\\
\qquad\quad~~ $\version(\snap_1)\geq \version(\snap_0)~\land$ \\
\qquad\quad~~ $\vercommit(\snap_1,\entries)~\land~\entry\notin \entries)$\\
~\\
\underline{$\oupdate(i, \snap,\pi)$}\\
if $(\clientupdate(\snaps[i], \snap,\pi))$\\
\quad $\snaps[i]\gets\snap$ \\
\quad add $\snap$ to $B$ \\
~\\
\underline{$\ocheck(i,\entry,\pi)$}\\
if $(\clientcheck(\snaps[i],\entry,\pi))$\\
\quad add $(\snaps[i],\entry,\pi)$ to $A$\\
\end{gamespec}
If for all PT adversaries $\advsplitview{\A}{\secp} < \nu(\secp)$ for some negligible function $\nu(\cdot)$ then the protocol \emph{resists split-view attacks}.
\end{definition}

Intuitively, this game maintains two lists: $A$ to keep track of the checkpoints that have been used to verify inclusion, and $B$ to keep track of the checkpoints that have been accepted as valid.  The adversary wins if one client has accepted a checkpoint $\snap_0$ and used it to verify the inclusion of an entry $\entry$ in the log ($(\snap_0,\entry,\pi)\in A$), but another client has accepted a checkpoint $\snap_1$ ($\snap_1\in B$) that commits to a set of entries that does not contain $\entry$; i.e., the two checkpoints are inconsistent.

\mariana{Are we going to be proving that a verifiable log construction defined with the six algorithms above is resistant to split-view attacks - if so, it seems that we should include some information/algorithms about gossip in the definition of a verifiable log (or at least I am not immediately seeing how the gossip will be included just in a construction instantiating the above definition)?} \sarah{Gossip will be folded into the $\clientupdate$ interactive protocol, which is a part of the definition too - the definition of a verifiable log is the six algorithms plus the two interactive protocols.}

\gary{Should $\clientcheck$ accept parameters? The arguments seem to match VerIncl. Alternatively, CheckEntry seems to imply additional client-server communication. Should $\ocheck$ be changed to VerIncl?}
}

We formally define security for a verifiable log in Appendix~\ref{sec:safety-proof}, in terms of \emph{append-only} security (which says it should be hard to provide a valid consistency proof for two checkpoints that commit to lists of entries where one is not a prefix of the other), \emph{membership} security (which says it should be hard to provide a valid inclusion proof for an entry not in the log), and the resistance to \emph{split-view attacks} (which says a server must present the same view of the log to every client).
If $\clientupdate$ and $\clientcheck$ are run only between the client and the server, then the adversary can always carry out a split-view attack by maintaining a different version of the log for each client but perform all operations on those individual logs honestly.  In order to ensure that clients actually have consistent views of the log, it is thus crucial for them to be aware of the checkpoints seen by other clients, in the form of a gossip protocol.  We explore the extra interactions this entails in $\clientupdate$ in Section~\ref{sec:gossip}.

\subsection{Verifiable registries}

We now define transparency for a dictionary $\dict$ that maps a key $\key$ to a value $\val$.  In fact, we keep track of not only the current mapping $\key\mapsto\val$ but all previous such mappings as well, and want to ensure that not only is the mapping correct but also the history of values is append-only.  Following Chase et al.~\cite{seemless}, 
we refer to this as a \emph{verifiable registry}.  We consider three types of client queries: (1) simple lookups, in which clients get the latest value for a key; (2) history lookups, in which they get the history of a key; and (3) audits, in which they ensure that the history of a key is append-only, meaning its value is consistent across all versions of the registry.  

More precisely, $\dict$ maps $\key$ to $\hist$, where $\hist$ contains both the current value (as its last entry) and all previous ones.  We denote by $\histrepr$ the concise representation of $\hist$ that is stored by the client.
We define the following algorithms run by the server, which holds as state the latest registry $\dict$:

\begin{sarahlist}

\item[$\snap,\dict\gets\append(\dict,\{\key_i,\val_i\}_i)$] appends some representation of $\val_i$ to $\dict[\key_i]$ for all $i$; i.e., adds $\val_i$ to the history associated with $\key_i$.  This results in a new version of the registry and thus a new checkpoint.

\item[$\val,\pi\gets\lookup(\dict,\key)$] returns the latest value $\val$ in $\hist\gets\dict[\key]$, along with a proof $\pi$ that $\val$ is correct.

\item[$\hist,\pi\gets\histlookup(\dict,\key,\histrepr_\old)$] returns the history $\hist$ of $\dict[\key]$ since $\histrepr_\old$ %
and a proof $\pi$ that $\hist$ is consistent with $\histrepr_\old$.

\item[$\hist,\pi\gets\audit(\dict,\key,\histrepr_\old)$] returns the history $\hist$ of $\dict[\key]$ since $\histrepr_\old$ and a proof $\pi$ that this is the \emph{complete} history consistent with it; i.e., that no inconsistent histories exist in any version of the registry, meaning $\dict[\key]$ is append-only.

\end{sarahlist}

We then consider the following verification algorithms run by the client, who stores a checkpoint $\snap$ (and, optionally, the history associated with one or multiple keys).  

\begin{sarahlist}

\item[$0/1\gets\verlookup(\snap,\key,\val,\pi)$] verifies the value $\val$ for $\key$ against $\snap$ and the proof $\pi$.  %

\item[$0/1\gets\verhistlookup(\snap,\key,\histrepr_\old,\hist,\pi)$] verifies the history $\hist$ for $\key$ against the stored history $\histrepr_\old$ and the proof $\pi$.  

\item[$0/1\gets\veraudit(\snap,\key,\histrepr_\old,\hist,\pi)$] verifies the history $\hist$ for $\key$ against the proof $\pi$, with respect to all of the registry versions since $\histrepr_\old$.  
\gary{Also might be worth noting that $\hist_\new$ should be select-able by the client (or at least derived from $\hist_\old$)}

\end{sarahlist}

In addition to these algorithms, a verifiable registry also requires the $\clientupdate$ interaction defined in Section~\ref{sec:gossip-defs}, in order to allow clients to update their checkpoints. 

We formally define security for a verifiable registry in Appendix~\ref{sec:registry-proof}, in terms of \emph{lookup security} (which says that it should be hard for a server to provide valid proofs for two different lookup values for the same key) and resistance to \emph{oscillation attacks}.  In such an attack, an adversarial server interacts with two types of clients, one performing basic lookups on a key and one performing full audits on the same key, and tries to ``oscillate'' between versions of the registry that show a consistent history to the auditing client and a value that isn't in this history to the more basic client.  If oscillation attacks are not possible for \emph{any} key, then the registry is \emph{append-only}~\cite{CCS:TBPPTD19}.

\newcommand\formalregistrydefns{
We again follow Tomescu et al.\ in defining correctness: $\verlookup(\snap,\key,\lookup(\dict,\key))=1$ and $\verhistlookup(\allowbreak\snap,\allowbreak\key,\histrepr,\allowbreak\pi)=1$ for $\pi\gets\histlookup(\dict,\allowbreak\key,\allowbreak\histrepr)$ when $\dict,\snap$ are produced by a series of $\append$ operations including $(\key,\val)$ and containing the combined history of $\histrepr$ and the history returned by $\histlookup$.  Similarly, we require that $\veraudit(\snap,\key,\histrepr,\audit(\dict,\key,\histrepr))=1$.
In terms of efficiency, registries should be able to support efficient lookups, meaning proofs that are of size $O(\log N)$ and take $O(\log N)$ time to verify.  Unlike logs, however, supporting efficient append-only proofs (which is essentially what $\audit$ proves) is more challenging for registries, as observed by Tomescu et al.~\cite{CCS:TBPPTD19} and as we elaborate on in Section~\ref{sec:maps}.  We defer a detailed discussion of efficiency for $\audit$ until Section~\ref{sec:registry-efficiency}.
In terms of security, we first consider the basic notion of \emph{lookup security}.

\begin{definition}{\emph{\textbf{\cite{CCS:TBPPTD19}}}}\label{def:lookup}
Define $\advlookup{\A}{\secp}$ as the probability of an adversary $\A$ outputting two different values that verify for the same key and registry version; i.e., $(\snap,\key,\val,\val',\pi,\pi')$ such that (1) $\val\neq\val'$, (2) $\verlookup(\snap,\key,\val,\pi) =1$, and (3) $\verlookup(\snap,\key,\val',\pi') =1$.  If for all PT adversaries $\advlookup{\A}{\secp} < \nu(\secp)$ for some negligible function $\nu(\secp)$, then the registry satisfies \emph{lookup security}.
\end{definition}

We also consider the more involved notion of an \emph{oscillation attack}.  In such an attack, an adversarial server interacts with two types of clients, one performing basic lookups on a key and one performing full audits on the same key, and tries to ``oscillate'' between versions of the registry that show a consistent history to the auditing client and a value that isn't in this history to the more basic client.  

\begin{definition}[Oscillation attack]\label{def:oscillation}
Define $\advoscillation{\A}{\secp}$ as the probability of an adversary $\A$ winning the following game (omitting the description of $\oupdate$, which is the same as in Definition~\ref{def:split-view}):
{\small
\begin{gamespec}
\underline{$\main$ $\oscillation{\A}{\secp}$}\\
$(\{\snap_i\}_{i=0}^2,\key,\val,\hist_\old,\hist,\pi_1,\pi_2)\randpick\A^{\oupdate}(\usecp)$\\
return $(\verlookup(\snap_1,\key,\val,\pi_1)~\land$\\
\qquad\quad~$\veraudit(\snap_2,\key,(\snap_0,\hist_\old),\hist,\pi_2)\land$\\
\qquad\quad~$\version(\snap_2)\geq \version(\snap_1)\geq\version(\snap_0)~\land$ \\
\qquad\quad~$\snap_0,\snap_1,\snap_2\in B~\land$ \\
\qquad\quad~$\val\neq\max(\hist_\old)~\land~\val\notin\hist$
\end{gamespec}
}
If for all PT adversaries $\advoscillation{\A}{\secp} < \nu(\secp)$ for some negligible function $\nu(\cdot)$ then the protocol \emph{resists oscillation attacks}.
\end{definition}

\mariana{The above definition for oscillation attack resistance seems to be the same as append-only, is there a difference? Do we need to define formally a version of the definition which is relevant just for the security of a single client?}

If oscillation attacks are not possible for \emph{any} key, then the registry is \emph{append-only}~\cite{CCS:TBPPTD19}.  This per-key variant is well-suited to applications such as Key Transparency, however, where participants don't necessarily care about the entire registry but just want to ensure that their own entries remain valid.
}

\section{Compact Ranges}\label{sec:compact}

In this section, we present the notion of a \emph{compact range}, which is a succinct representation of a range of leaves in a history tree (Section~\ref{sec:history-trees}).  Both inclusion proofs and consistency proofs can be represented using compact ranges, which enable several optimizations that we take advantage of in both our gossip protocol and our verifiable registry.

In a Merkle tree, we define a \emph{range} of consecutive leaves as $\range = [L,R)$, where $L$ is the index of the leftmost leaf and $R-1$ is the index of the rightmost one.
We call the \emph{coverage} of a node $\node$ the set of leaves that are descendants of $\node$; we denote this by $\coverage(\node)$.  Similarly, the coverage of a set of nodes $S$ is $\cup_{\node\in S} \coverage(\node)$.  
We call a set $S$ the \emph{covering set} of a range $\range$ if $\coverage(S) = \range$.

\begin{definition}[Compact range]
The \emph{compact range} $\compact{\range}$ of a range $\range$ is the smallest covering set $S$ such that each node $\node\in S$ covers exactly $2^{\hgt(\node)}$ leaves; i.e., $\compact{\range} = \min\{S~:~\coverage(S) = \range~\land~ |\coverage(\node)| = 2^{\hgt(\node)}~\forall~\node\in S\}$.
\end{definition}

We consider a recursive algorithm for computing a compact range, as specified in Algorithm~\ref{alg:compact-range} in Appendix~\ref{sec:compact-app} (most of our pseudocode is in this appendix).
It constructs the compact range recursively and in a greedy manner: going down both left and right from the root, it tries to take nodes highest up in the tree, which cover the largest number of leaves.  If a node's coverage does not overlap with the range $[L,R)$, the algorithm does not continue with this branch.  If the coverage is too broad or the node's subtree is not perfect, the algorithm continues to its children.  If the coverage fits within the range, the algorithm adds it to the compact range.  The runtime of \compactrangealg is $O(\log N)$ and $|\compact{\range}| = O(\log N)$ as well.  This is because, due to the greedy nature of the algorithm, a node is added only if its parent is not added.  There are two reasons why this can happen: (1) this is a right child, and the parent's coverage begins before $L$, and (2) this is a left child, and the parent's coverage ends after $R$.  Each of these situations can happen for only one node per height, so the algorithm takes at most two nodes per height.  Since there are at most $\log(R-L)+1$ heights that can have a perfect subtree fitting within the range $[L,R)$, $|\compact{\range}| \leq 2\cdot(\log(R-L)+1) = O(\log N)$.  
In some cases the compact range can be smaller, as shown in Figure~\ref{fig:pretty-compact-ranges}.  
In the most extreme case, the range of all leaves in a tree of size $2^n$ can be represented by just the root.

\newcommand\compactalg{
\begin{algorithm}[t]
\KwIn{A node $\node$ in a Merkle tree of size $N$ (initially set to $\croot$), and a range $\range = [L,R)$ such that $0\leq L\leq R \leq N$}
$\mathsf{left} \gets \min(\coverage(\node))$ \\
$\mathsf{right} \gets \mathsf{left}+2^{\hgt(\node)}$ \\
\uIf{$\mathsf{left} \geq L$ \emph{\textbf{and}} $\mathsf{right} \leq R$}{
\KwRet $[\node]$ \\
}
\uElseIf{$\mathsf{left} \geq R$ \emph{\textbf{or}} $\mathsf{right} \leq L$}{
\KwRet $[~]$ \\
}
\uElse{
$d_\ell\gets\compactrangealg(\node[\ell],L,R)$ \\
$d_r\gets\compactrangealg(\node[r],L,R)$ \\
\KwRet $d_\ell \| d_r$ \\
}
\caption{\compactrangealg}
\label{alg:compact-range}
\end{algorithm}
}

\begin{figure}
\centering
\includegraphics[width=\linewidth]{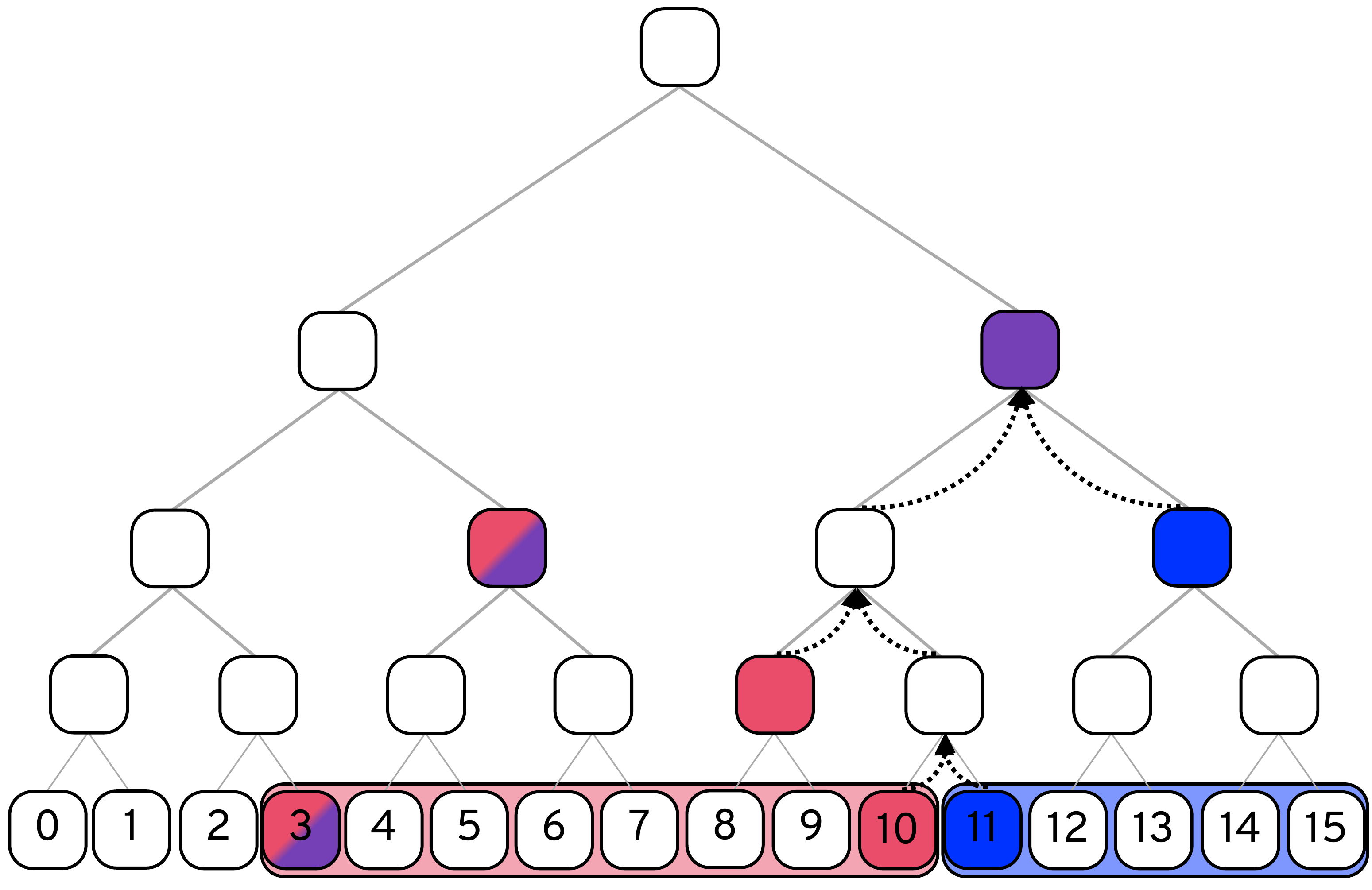}
\caption{In a tree of size 16, the compact ranges $\compact{[3,11)}$ (in red) and $\compact{[11,16)}$ (in blue).  The nodes in purple represent the merged compact range $\compact{[3,16)}$ and the dotted black lines represent the merging operations in Algorithm~\ref{alg:compact-merge}.  The two nodes that are both red and purple are included in both $\compact{[3,11)}$ and $\compact{[3,16)}$.}
\label{fig:pretty-compact-ranges}
\end{figure}

Beyond serving as a succinct representation of a sequence of leaves, compact ranges can be efficiently \emph{merged}, using \mergealg (Algorithm~\ref{alg:compact-merge}).  This algorithm starts by taking the union of the nodes in multiple ranges, and then looks for nodes with the same parent, starting at the lowest height.  If it finds any such pairs, it replaces them with their parent.  %
These pairs can be found only along the edges of each compact range (e.g., 
the rightmost node in the left range and the leftmost node in the right range), since otherwise they would not have been added to the original range.  This means they can be identified in constant time at every step, so the runtime of \mergealg is $O(\log N)$ and the merged compact range is also of size $O(\log N)$.  Beyond asymptotics, merging two compact ranges $\compact{[L_1,R_1)}$ and $\compact{[L_2,R_2)}$ is most impactful in the case that $L_2 = R_1$ (i.e., the two ranges share a border), since we can expect to have more nodes with shared parents here, as depicted in Figure~\ref{fig:pretty-compact-ranges}.

\newcommand\compactmergealg{
\begin{algorithm}[t]
\KwIn{Compact ranges $\{\compact{[L_i,R_i)}\}_i$}
$\mathsf{merged}\gets \bigcup_i \compact{[L_i,R_i)}_i$ \\
\While{$\exists~\node_1\neq \node_2\in\mathsf{merged}~:~\node_1[\parent] = \node_2[\parent]$}{
remove $\node_1,\node_2$ from $\mathsf{merged}$ \\
add $\node_1[\parent]$ to $\mathsf{merged}$ \\
}
\KwRet $\mathsf{merged}$ \\
\pavel{maybe elaborate that adding a parent also means computing its hash from the children v1 and v2}
\caption{\mergealg}
\label{alg:compact-merge}
\end{algorithm}
}

\newcommand\compactrangerootalg{
\begin{algorithm}[t]
\KwIn{A compact range $\compact{[0,N)}$ for $N > 0$}
\KwOut{The root hash $\roothash$ of a tree of size $N$}
$c\gets\textsc{SortedHashStack}(\compact{[0,N)},\mathsf{key}=\hgt)$ \\
\While{$|c| > 1$}{
$c_r\gets c.\mathsf{pop}()$ \\
$c_\ell\gets c.\mathsf{pop}()$ \\
$h\gets H(c_\ell \| c_r)$ \\
$c.\mathsf{push}(h)$ \\
}
\KwRet $c[0]$ \\
\caption{\rangerootalg}
\label{alg:range-to-root}
\end{algorithm}
}

If we have a compact range for $\range=[0,\ell)$ then the nodes in this compact range allow us to compute the root of the history tree of size $\ell$, using \rangerootalg as shown in Algorithm~\ref{alg:range-to-root} and proved in Lemma~\ref{lem:range-to-root}.

\newcommand\lemmarangeroot{
\begin{lemma}\label{lem:range-to-root}
For all $\ell\in\N$, \rangerootalg (Algorithm~\ref{alg:range-to-root}) correctly computes the root of the history tree of size $\ell$, given the nodes in $\compact{[0,\ell)}$.
\end{lemma}

\begin{proof}
First, if $\ell = 2^k$ for some value $k$, then the compact range contains one node, which is the root of the history tree of size $\ell$ (because the shape of a history tree is uniquely defined by its number of leaves).  Its hash is thus by definition the root hash, meaning \rangerootalg is successful.

More generally, the left subtree of a history tree is (by definition) of size $2^i$, where $i = \lfloor\log(\ell-1)\rfloor$.  We can use the recursive formula for computing Merkle hashes to define $\roothash = \mth(0,\ell) = H(\mth(0,2^i) \| \mth(2^i,\ell))$, where $\mth(i_1,i_2)$ is the Merkle hash of the data contained between the $i_1$-th and $i_2$-th leaves~\cite[Section 2.1]{6962}.  The right subtree is itself a history tree, meaning we can further decompose $\mth(2^i,\ell) = H(\mth(2^i,2^i + 2^{i_2})\| \mth(2^i + 2^{i_2},\ell))$, where $i_2 = \lfloor \log(\ell - 2^i - 1)\rfloor$.  This means that if we create a stack using $\mth(0,2^i)$, $\mth(2^i,2^i + 2^{i_2})$, and $\mth(2^i + 2^{i_2},\ell)$ (pushing values in that order) then the steps in Algorithm~\ref{alg:range-to-root} correctly compute the root, which is $H(\mth(0,2^i) \| H(\mth(2^i,2^i + 2^{i_2}) \| \mth(2^i + 2^{i_2},\ell)))$.  If we fully decompose $\ell = 2^{i_1} + 2^{i_2} + \ldots + 2^{i_j}$, where $i_1 > \cdots > i_j$ then (1) the greedy nature of compact ranges (Algorithm~\ref{alg:compact-range}) and (2) the fact that complete subtrees are ``frozen'' in history trees of every size mean that the $k$-th element of the compact range exactly represents $\mth(S_k, S_k + 2^{i_k})$, where $S_1 = 0$ and $S_k = \sum_{m=1}^{k-1} 2^{i_m}$ for all $k$, $1\leq k\leq j$.  In other words, the hashes that are part of the compact range are exactly the subtree hashes necessary to compute the root.
\end{proof}
}

As shown in Algorithm~\ref{alg:ver-incl}, the combined ability to merge ranges and compute roots from them means we can formulate inclusion proofs in terms of compact ranges.  In particular, for a tree of size $N$ a proof of inclusion of an entry at the $i$-th leaf can consist of $\compact{[0,i)}$ and $\compact{[i+1,N)}$.  A verifier in possession of a checkpoint of the form $(\roothash,N)$ can then merge these ranges with the entry itself to get $\compact{[0,N)}$, and from that compute the root hash of the tree of size $N$ and compare it to $\roothash$.  This approach also generalizes to proving inclusion of a range of consecutive leaf values, in which case the verifier takes in multiple consecutive entries instead of just one.  Even in the case of non-consecutive leaf values at indices $i_1,\ldots,i_k$, we can think of an inclusion proof as consisting of $\{\compact{[i_j+1,i_{j+1})}\}_{j=0}^k$, where $i_0 = -1$ and $i_{k+1} = N$.

\begin{algorithm}[t]
\KwIn{A checkpoint $(\roothash,N)$, leaf $\mathsf{L}$, and proof $\pi$}
$(\compact{\range_1},\compact{\range_2})\gets\pi$ \\
$\mathsf{merged}\gets\mergealg(\compact{\range_1},\mathsf{L},\compact{\range_2})$ \\
$\roothash_m\gets\rangerootalg(\mathsf{merged})$ \\
\uIf{$\roothash_m = \roothash$}{
\KwRet 1 \\
}
\Else {
\KwRet 0 \\
}
\caption{$\verincl$}
\label{alg:ver-incl}
\end{algorithm}

Similarly, we can formulate consistency proofs in terms of compact ranges, as shown in Algorithm~\ref{alg:ver-append}.  This is because a consistency proof from $(\roothash,\ell)$ to $(\roothash_\new,\ell_\new)$ is essentially a combination of (1) a set of nodes that allows us to compute $\roothash$ within the tree of size $\ell$ and (2) a set of nodes that allows us to compute $\roothash_\new$ incrementally from this first set of nodes.  This first set is just the compact range $\compact{[0,\ell)}$, because we can use this to compute the root hash of the tree of size $\ell$ and compare it to $\roothash$.  Similarly, the second set is just the compact range $\compact{[\ell,\ell_\new)}$, because we can merge this with the first range to form $\compact{[0,\ell_\new)}$ and from this compute the root hash of the tree of size $\ell_\new$ and compare it to $\roothash_\new$.

\begin{algorithm}[t]
\KwIn{Checkpoints $(\roothash,\ell)$ and $(\roothash_\new,\ell_\new)$, and a proof $\pi$}
$(\compact{\range_1},\compact{\range_2})\gets\pi$ \\
$\roothash_1 \gets \rangerootalg(\compact{\range_1})$ \\
\uIf{$\roothash_1 = \roothash$}{
$\mathsf{merged}\gets\mergealg(\compact{\range_1},\compact{\range_2})$ \\
$\roothash_2 \gets \rangerootalg(\mathsf{merged})$ \\
\uIf{$\roothash_2 = \roothash_\new$}{
\KwRet 1 \\
}
\Else {
\KwRet 0 \\
}
}
\Else{ 
\KwRet 0 \\
}
\caption{$\verappend$}
\label{alg:ver-append}
\end{algorithm}

In the typical formulation, clients check consistency proofs and store only the new checkpoint containing $(\roothash_\new,\ell_\new)$.  Instead, clients who expect to update their checkpoints often could store $\compact{[0,\ell_\new)} = \textsc{Merge}(\compact{[0,\ell)}, \compact{[\ell,\ell_\new)})$ and then ask for (and verify) only the second compact range on their next update.  This imposes a higher storage requirement, since $\compact{[0,\ell_\new)}$ is $O(\log(\ell_\new))$, but allows them to use less bandwidth and perform less computation, since the second set of nodes is of size at most $2\log(\ell_\new-\ell)$.  Assuming large logs and relatively frequent updates (on the order of hours or days rather than months or years), this means consistency proofs can be considered constant-sized, as we explore in Section~\ref{sec:gossip-performance}.

\section{Gossiping about Verifiable Logs}\label{sec:gossip}

As argued in Section~\ref{sec:gossip-defs}, detecting split-view attacks is not possible if clients do not know about the checkpoints visible to other clients, which means they need to hear about them from someone other than the (untrusted) server.  Consequently, they either need to communicate with other clients (which, as argued in the introduction, is hard to implement and not scalable), or another participant needs to play this role~\cite{sundr}.  

Our proposed protocol thus introduces a new participant called a \emph{witness}~\cite{SP:STVWJG16}, although this role could be played by other log servers without affecting security.  To make it as easy as possible to recruit new participants, our goal is to minimize the requirements placed on these witnesses in terms of their uptime, interaction with each other, and computational and storage costs.  Even without interaction between witnesses, our protocol is able to achieve a loose form of consensus amongst them, which means we can provably prevent split-view attacks rather than just detect them retrospectively.

\subsection{Our proposed protocol}

Our proposed gossip protocol considers an additional type of participant called a witness, who is involved in the $\clientupdate$ interaction introduced in Section~\ref{sec:gossip-defs}.  Each witness stores a list of checkpoints $\snaps$ associated with a given server.  A checkpoint is of the form $\snap = (\roothash,\ell,t,\sigma)$, where $\roothash$ is a commitment to a log of size $\ell$ (e.g., the root hash of a Merkle tree), $t$ is a timestamp, and $\sigma$ is a signature over $(\roothash,\ell,t)$.  This means that $\version(\snap) = \ell$.  To ensure that witnesses can keep their storage costs fixed, we consider that $\snaps$ has a fixed maximum length.  If a witness attempts to add a new checkpoint to a full list, it first removes an existing checkpoint according to some eviction strategy (e.g., it removes the oldest checkpoint, according to $\ell$ and then $t$).

Our new $\clientupdate$ protocol proceeds in two phases: first, each server \emph{broadcasts} its latest checkpoint to the set of $\numwitnesses$ witnesses, who store it if they are online and verify it as being consistent with previous checkpoints.  Then, clients \emph{collect} checkpoints from each of the witnesses in the hopes of finding at least one checkpoint that a sufficient majority $Q$ agree is valid.  We assume that all participants are aware of each server's public key $\pk_S$ and each witness' public key $\pk_W$ for a digital signature scheme $(\keygen,\sign,\verify)$.
\gary{Do we have a formal definition for consistency?} 
We begin by describing the broadcast phase.

\begin{sarahlist}
\item[B1]  A server broadcasts a new checkpoint $\snap_\new$ to the witnesses.

\gary{Do we need to describe a registration protocol previous to this by which servers can become aware of witnesses $\numwitnesses$?}

\item[B2] If a witness is online and receives $\snap_\new = (\roothash,\ell,t,\sigma)$, it first checks that $\verify(\pk_S,(\roothash,\ell,t),\sigma) = 1$ to ensure that this is a valid checkpoint for this server.  If this passes, it sends back its freshest stored checkpoint $\snap$ (i.e., the checkpoint with the highest version) in order to request a consistency proof. 
\gary{$\verify$ does not appear in the operations of the Verifiable Log. Should it be added? I think this is the first place we expand the definition of $\snap$ to a tuple.}
\gary{Is it worth noting that this is one possible strategy of many for acquiring consistency proofs?}

\item[B3] The server forms $\pi\gets\proveappend(\mlog,\snap,\allowbreak\snap_\new)$ and sends this back.

\item[B4] If $\verappend(\snap,\snap_\new,\pi) = 1$, the witness adds $\snap_\new$ to its list.  
If the proof doesn't verify, the witness can send $(\snap,\snap_\new,\pi)$ to an auditor for further investigation. 
\end{sarahlist}
\gary{Is one invalid proof sufficient to prove malicious intent, or is it the absence of an available proof for a certain amount of time, or will the absence of a proof automatically ensure consensus is not reached?}

Next, we describe the collection phase.

\begin{sarahlist}
\item[C1] A client broadcasts a request to the witnesses for a list of $\numreqsnap$ checkpoints, according to some request policy (e.g., the $\numreqsnap$ most recent checkpoints for a given server).

\item[C2] If a witness is online and receives this request, it forms this list and forms $\sigma_i\randpick\sign(\sk_W,\snap_i)$ for every $\snap_i$ in it.  It then responds with the list $\{\snap_i,\sigma_i\}_{i=1}^{\numreqsnap}$. 

\item[C3] Upon receiving $\{\snap_i,\sigma_i\}_{i=1}^{\numreqsnap}$, the client first parses $(\roothash_i,\ell_i,t_i,\sigma_i')\gets\snap_i$ and checks that $\verify(\pk_W,\allowbreak\snap_i,\sigma_i)=1$ and $\verify(\pk_S,(\roothash_i,\ell_i,t_i),\sigma_i') = 1$ for all $i$, $1\leq i\leq \numreqsnap$.  
After hearing from $\mu$ witnesses, the client ends up with a set of checkpoints $\{\snap_{i,j}\}_{i,j=1}^{\numreqsnap,\mu}$.  It then accepts as valid all checkpoints that it received from at least $Q$ witnesses, and updates its local state with the freshest valid checkpoint.  If there is no consensus on any checkpoint, it repeats the process starting from Step~C1.
\end{sarahlist}

\subsection{Security}

To prove this protocol secure, we consider that checkpoints accepted by honest clients fall within a sort of consensus view, as they must have been seen by at least $Q$ witnesses.  
As such, we want to argue that the protocol achieves the standard notions of \emph{safety}, meaning honest clients accept only ``good'' checkpoints, and \emph{liveness}, meaning honest clients can continue to update their checkpoints.  %

We allow up to $F$ of the $\numwitnesses$ witnesses to be adversarial, where we define $\numwitnesses = VF +1$ (e.g., traditional consensus protocols often consider $\numwitnesses = 3F+1$, or $V=3$).  

\begin{lemma}\label{lem:bundle-size}
If the signature scheme is unforgeable and a client receives valid signatures from at least $Q=\frac{(V+1)\cdot F}{2} + 1$ witnesses on two distinct checkpoints $\snap$ and $\snap'$, where the total number of witnesses is $\numwitnesses = VF + 1$, then there exists at least one honest witness who has signed both $\snap$ and $\snap'$.
\end{lemma}

\begin{proof}
If the client has two sets of valid signatures, where each set is of this size $Q$, then they have 
\begin{align*}
2((V+1)\cdot F/2 + 1) &= (V+1)\cdot F +2 \\
&= (VF + 1) + F + 1\\
&= \numwitnesses + F + 1
\end{align*} 
valid signatures.  Since there are only $\numwitnesses$ witnesses in total, this means that $F + 1$ witnesses must have signed both $\snap$ and $\snap'$.  Assuming signature unforgeability, which means an adversarial witness cannot form a valid signature that looks like it came from an honest one, this further implies that at least one honest witness signed both checkpoints.
\end{proof}

If $\numwitnesses = 3F + 1$, this means that clients need to see signatures from at least $2F+1$ witnesses; alternatively, if they require signatures from every witness then we can tolerate having all but one witness be adversarial.%

To further consider liveness, we need to acknowledge that not all witnesses may be online at all times.  In particular, we need to consider the \emph{uptime} $U$ of each witness, in terms of the probability that it is online at a given point in time.  Most consensus protocols operate in a \emph{partially synchronous} network model~\cite{dln88}, meaning messages between honest participants are delivered within some bound $\Delta$, rather than an \emph{asynchronous} model~\cite{ckps01} where an adversary controls the delivery of all messages and thus they can be delayed arbitrarily.  Our consideration of uptime is more akin to the ``sleepy'' model due to Pass and Shi~\cite{sleepy}, in which all honest nodes have access to a weakly synchronized clock (meaning messages are delivered within some bounded delay) but may periodically go offline or be ``asleep.''  Unlike Pass and Shi, we do not assume that a sleepy node receives all previously delivered messages when they wake up; instead, we consider these messages as lost forever.

\begin{theorem}\label{thm:gossip-liveness}
If (1) the minimum uptime $U$ of any honest witness is at least $\frac{(V+3)F + 2}{2\numwitnesses}$, (2) all online honest participants receive messages from other online honest participants within some delay $\Delta$, and (3) $Q$ is defined as in Lemma~\ref{lem:bundle-size}, then an honest client who has accepted a checkpoint $\snap$ will eventually accept another checkpoint $\snap_\new$ such that $\version(\snap_\new) > \version(\snap)$.
\end{theorem}

\begin{proof}
Without loss of generality, we assume that adversarial witnesses simply ignore client requests.  By our assumptions about online parties, all client requests will eventually reach all online honest witnesses and their responses will eventually reach the client.  In order for a client to accept $\snap_\new$, they must see signatures on it from at least $Q$ witnesses.  This means that, of the witnesses online at the time of the client's request (of which there are $U\cdot \numwitnesses$), at least $Q$ of them must be honest (as again, adversarial witnesses just ignore the request).  This means we need $U\cdot \numwitnesses- F \geq Q$, or 
\begin{align*}
U &\geq \frac{Q+F}{\numwitnesses} \\
&\geq \frac{\frac{(V+1)\cdot F}{2} + 1 + F}{\numwitnesses} \\
&\geq \frac{(V+3)F + 2}{2\numwitnesses},
\end{align*}
as desired.
\end{proof}

If $\numwitnesses= 3F+1$ then this theorem says that each witness must be available all the time; i.e., to have $U = 1$.  
This is in line with existing consensus protocols that allow for one-third of the participants to be adversarial but require all messages to be delivered eventually.  
We do not view perfect uptime as a realistic requirement, however, so explore the effect of different uptimes on performance and liveness in Section~\ref{sec:gossip-performance}.

We now prove that an honest client is guaranteed \emph{safety} if they accept only checkpoints that have been seen by $Q$ witnesses, as long as an honest \emph{mirror} has done the same.  We define mirrors as clients that maintain their own copy of the log entries.  
Mirrors can thus be used as ground-truth information about the log entries, which is needed in our proof of security.
\gary{Does this imply that clients that only accept. checkpoints that have been seen by $Q$ witnesses are not safe?}

\begin{theorem}\label{thm:gossip-safety}
If the hash function is collision-resistant and append-only security and membership security hold for the log (Definitions~\ref{def:append-only-log} and~\ref{def:membership-log}), then no split-view attack is possible between a client and a mirror who ran $\clientupdate$ with the same set of witnesses.
\end{theorem}

\newcommand\safetyproof{
\begin{proof}
Let $\A$ be an adversary playing game $\splitview{\A}{\secp}$.  We build PT adversaries $\B_1$, $\B_2$, and $\B_3$ such that
\[
\advsplitview{\A}{\secp} \leq \advcr{\B_1}{\secp} + \advappend{\B_2}{\secp} + \advmemb{\B_3}{\secp},
\]
where $\advcr{\B_1}{\secp}$ denotes the probability of $\B_1$ outputting values $x_1$ and $x_2$ such that $H(x_1) = H(x_2)$.  We consider the winning conditions for $\splitview{\A}{\secp}$.  If $\A$ succeeds, then it outputs $\entries$ and there exist $(i,\snap,\entry,\pi)$ and $(j,\snap')$ such that (1) $\verincl(\snap,\entry,\pi)=1$ (by definition of the honest client), and (2) the honest client accepts $\snap$ and the honest mirror accepts $\snap'$; i.e., each of them ended the $\clientupdate$ interaction by changing their state to be this checkpoint.  Furthermore, the other winning conditions require that (3) $\version(\snap') \geq \version(\snap)$, (4) $\vercommit(\snap',\entries) = 1$, and (5) $\entry\notin\entries$.

We first consider the case that $\entries'\neq\entries$ for the list $\entries'$ maintained by the mirror, even though $\vercommit(\snap',\entries') = 1$ (by definition of the mirror) and $\vercommit(\snap',\entries) = 1$ (by the fourth condition).  If this were the case, meaning there were an index $k$ such that $\entries'[k]\neq\entries[k]$, there would be a node in the log that represents two distinct descendants: $\entries'[k]$ and $\entries[k]$.  This further implies that there are values $x$, $x'$, $y$, and $y'$ such that either $x\neq x'$ or $y\neq y'$, and $H(x\| y) = H(x'\| y')$, which means we can construct $\B_1$ to break collision resistance by outputting $x\| y$ and $x'\| y'$.  We thus assume for the remainder of the argument that $\entries' = \entries$.

By the second condition, the honest client and mirror respectively accept checkpoints $\snap$ and $\snap'$.  This means they must be stored in the list of at least $Q$ witnesses.  By Lemma~\ref{lem:bundle-size}, there thus exists at least one honest witness who has seen a path of valid consistency proofs $\pi_i$ from $\snap$ to $\snap'$ (by the third condition), meaning a series of checkpoints and proofs $\{\snap_i,\pi_i\}_{i=1}^n$ such that $\verappend(\snap_{i-1}, \snap_i, \pi_i) = 1$, where $\snap_0 = \snap$ and $\snap_n = \snap'$.  
For every $k$, there must exist a prefix $\entries_k$ of $\entries$ (because $\entries = \entries'$) such that $\vercommit(\snap_k,\entries_k)$; i.e., a list $\entries[1 : j_k]$ for which $\snap_k$ is a valid checkpoint.  If at some step $k$ there doesn't exist such a prefix, then we could construct an adversary $\B_2$ to break append-only security by outputting $(\snap_k, \snap_{k+1}, \entries[1 : j_{k+1}], \pi_k)$.  If this holds for all $k$, then there exists a list $\entries_c$ that is a prefix of $\entries$ and such that $\vercommit(\snap,\entries_c)=1$.

By the first condition, the client has seen a proof $\pi$ such that $\verincl(\snap,\entry,\pi) = 1$, but by the fifth condition $\entry\notin\entries$.  This further implies that $\entry\notin\entries_c$, which means we can construct an adversary $\B_3$ to break membership security by outputting $(\snap,\entry,\entries_c,\pi)$.
\end{proof}
}

A proof of this theorem can be found in Appendix~\ref{sec:safety-proof}.  
As a corollary, if multiple clients and a mirror all run $\clientupdate$ with the same set of witnesses, then no split-view attack is possible between any of the clients either.  This raises the question of how often the set of witnesses might change, or how the protocol might tolerate changes in witnesses.  We leave the latter question as future research, but observe that if the role of witnesses is played by existing servers in a setting like CT, this is a relatively static list that is furthermore distributed in an authoritative way,\footnote{\url{https://certificate.transparency.dev/logs/}} 
meaning clients can be sure that they all have the same list of witnesses.

Finally, we observe that while the threshold $Q$ must be at least as high as defined in Lemma~\ref{lem:bundle-size} in order to prove safety, some clients may want to use other policies.  For example, if clients in CT represent browser vendors (rather than individual browsers), the Chrome client might want to ensure that one of the $Q$ witnesses is run by Google.  Our protocol is flexible in this regard and provides clients with full discretion over which checkpoints they accept.

\subsection{Implementation and evaluation}
\label{sec:gossip-performance}

We implemented our gossip protocol in Go and ran simulations on a Linux workstation with an Intel Xeon CPU W-2135 (3.70\si{\giga\hertz}) and 65\si{\giga\byte} of RAM. Each server, witness, and client is run as a separate Goroutine, and a latency of 100\si{\milli\second} is imposed on communication between all entities to simulate a globally distributed deployment.  For the cryptographic operations, we draw on our own open-source implementation of compact ranges and digital signatures.\footnote{\ifsubmission\else{\url{https://github.com/google/trillian/tree/master/merkle/compact}}\fi}

Our experiments are parameterized by five values: (1) the number of servers; (2) the number of witnesses $\numwitnesses$; (3) the number of clients; (4) the fraction of adversarial witnesses $V$; and (5) the minimum uptime $U$.
We chose to have each witness not evict any checkpoints, as a checkpoint is only 112~bytes (a 32-byte hash $h$, an 8-byte integer $\ell$, an 8-byte timestamp $t$, and a 64-byte signature $\sigma$).  It thus requires only on the order of hundreds of megabytes to store millions of checkpoints.

We consider three different configurations of these parameters, which we summarize in Table~\ref{tab:params} in Appendix~\ref{sec:params}.  Our ``aggressive'' setting captures a general transparency deployment, in which there are a large number of servers and clients but a small number of witnesses who support gossip for many different use cases.  We allow many (25\%) of these witnesses to be adversarial and impose almost the lowest possible uptime (90\%) given the constraints imposed by Theorem~\ref{thm:gossip-liveness}.  Our ``KT'' setting captures having a large number of clients (representing individual users) and a higher number of witnesses.  These witnesses have a low uptime requirement (85\%) and fewer are assumed to be adversarial (12.5\%).  Finally, our ``CT'' setting captures CT as it is deployed today, meaning having a small number of servers and witnesses, and a small number of clients given that clients are more likely to represent browser vendors than individual browsers.  We also expect more established organizations to act as witnesses (as discussed, this role could even be played by other servers), meaning we can impose a high minimum uptime (99\%) and assume a relatively low fraction of adversarial witnesses (12.5\%).

\subsubsection{Microbenchmarks}

There are four main operations in our gossip protocol: $\proveappend$, $\verappend$, $\sign$, and $\verify$.  For the first two, we use the optimization described at the end of Section~\ref{sec:compact}, in which witnesses with a checkpoint containing $(h,\ell)$ store $\compact{[0,\ell)}$ and ask only for the difference $\compact{[\ell,\ell_\new)}$ when updating to a new checkpoint containing $(h_\new,\ell_new)$.  We consider logs of size $2^{10}$ and $2^{26}$ in what follows, and consider that $\ell_\new-\ell$ is bounded by $30$.  These choices are inspired by how gossip works for the \mogname registry in Section~\ref{sec:registry}, and we justify them there.  
Averaged over thousands of runs, $\proveappend$ for a log of size $2^{26}$ and a difference of $30$ leaves took $6.04$\si{\micro\second}, $\verappend$ took $6.22$\si{\micro\second}, $\sign$ took $89.89$\si{\micro\second}, and $\verify$ took $184.17$\si{\micro\second}.

\subsubsection{Liveness}\label{sec:gossip-liveness}

In order to determine a client's ability to move their checkpoints forward, we ran the collection phase with the parameters defined by each of our three settings and with $\numreqsnap$ ranging from $1$ to $30$.  For each value, we had each client request the $\numreqsnap$ most recent checkpoints for each server, and then counted the fraction of servers for which the collective set of clients were, on average, able to find a new valid checkpoint after one request.  The results are plotted in Figure~\ref{fig:num-checkpoints}.  

\begin{figure}
\centering
\includegraphics[width=0.75\linewidth]{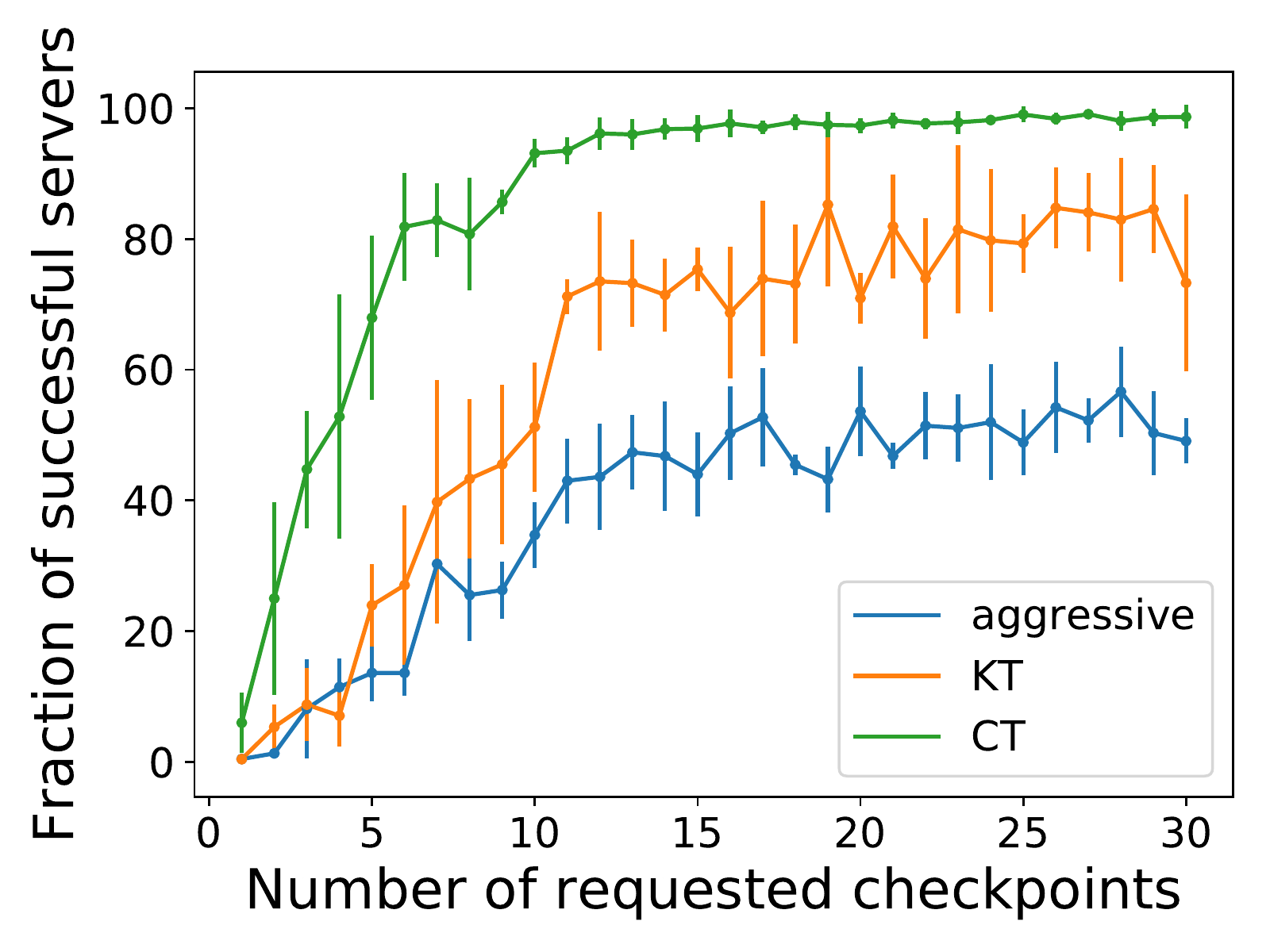}
\caption{The percentage of servers for which the average client was able to find consensus amongst $Q$ witnesses after one request for the $\numreqsnap$ most recent checkpoints per server, for each of our three settings.  The average was taken across all clients in a run and across five runs.}
\label{fig:num-checkpoints}
\end{figure}

Unsurprisingly, it is easiest for clients to get near-full consensus in the CT setting, in which witnesses have the highest uptime (99\%).  Despite the higher uptime in the aggressive setting than the KT setting, %
the fact that there are fewer honest witnesses means that clients are more likely to just get unlucky and not reach enough witnesses.  This suggests that in settings with low required uptime it is important to have more witnesses to provide the same level of coverage.
\gary{is there an equation that expresses this relationship? i.e. If an application wants a particular success probability (99\%), how many witnesses would be needed?}

\subsubsection{Latency and bandwidth}\label{sec:gossip-latency}

We consider the end-to-end runtime of, and bandwidth required by, both our broadcast and collection phases, according to the parameters defined by each of our settings.  
To pick the value of $\numreqsnap$ in the collection phase, we chose the minimal value for which clients found consensus for 50\% of servers: $16$ for the aggressive setting, $10$ for the KT setting, and $4$ for the CT setting.  We also consider the minimal value to achieve a threshold we consider more realistic for each setting.  In the aggressive setting, different servers may support completely different use cases, so there may be many servers with which clients never interact.  We thus set a threshold of 10\% here, which means $\numreqsnap=4$.  Similarly, in the KT setting different servers may contain keys for different messaging applications and thus clients may update their checkpoints at different points in time, so we set a threshold of 25\%, or $\numreqsnap=6$.  Finally, in the CT setting it is more likely that clients would want to update their checkpoints for many servers at the same time, so we set a threshold of 90\%, or $\numreqsnap=10$.
\gary{Is the fraction a percentage of servers also the likely hood of the protocol completing successfully for 1 server?}

We then measure the end-to-end latency in both phases, in terms of the total time (and bandwidth) required.  In the broadcast phase this is from the point at which a single checkpoint is broadcast (by a single server) to the point at which it is accepted by $Q$ witnesses, and in the collection phase this is from the point at which a client broadcasts a request for a single server to the point at which it finds a new valid checkpoint for that server.  The results are in Table~\ref{tab:latency}.

\begin{table}[t]
\centering
\setlength\tabcolsep{5pt}
{\footnotesize
\begin{tabular}{llccS[table-format=2.0]S[table-format=4.0]S[table-format=3.2]}
\toprule
Phase & Setting & Log size & $T$ & {$\numreqsnap$} & {Time (\si{\milli\second})} & {Size (\si{\kilo\byte})} \\
\midrule
B & Aggressive & $2^{26}$ & -- & {--} & 3060 & 2.2 \\
B & KT & $2^{26}$ & -- & {--} & 300 & 0.15 \\
B & CT & $2^{10}$ & -- & {--} & 480 & 0.25 \\
C & Aggressive & -- & 50 & 16 & 200 & 65.17\\
C & KT & -- & 50 & 10 & 202 & 145.23 \\
C & CT & -- & 50 & 4 & 200 & 17.46 \\
C & Aggressive & -- & 10 & 4 & 201 & 16.63 \\
C & KT & -- & 25 & 6 & 201 & 16.63 \\
C & CT & -- & 90 & 10 & 202 & 43.61 \\
\bottomrule
\end{tabular}
}
\caption{The average end-to-end runtime and required bandwidth for each of the two phases of our gossip protocol (B = Broadcast and C = Collection), in each of the settings we consider and averaged across five runs.  We use $T$ to denote the percentage of servers for which the client requests new checkpoints in the collection phase, which determines the number of requested checkpoints $\numreqsnap$.}
\label{tab:latency}
\end{table}

In both the KT and CT settings, checkpoints can be gossiped in less than a second.  In the aggressive setting, the low number of witnesses and low uptime requirement means that a server may be unable to get its checkpoint accepted by a sufficient threshold of witnesses on the first try.  Our runtime of 3.1\si{\second} thus reflects the server repeating its request until this happens.  In all settings, the runtime is completely dominated by network latency.  
The runtime in the collection phase is also dominated by network latency, and barely changes across the settings.  This is expected, as $\numreqsnap$ was chosen to ensure that clients would not have to repeat their request and for such low numbers of $\numreqsnap$ computation does not meaningfully change and is on the order of microseconds rather than milliseconds.  Bandwidth, on the other hand, grows linearly in $\numreqsnap$ and the number of witnesses, with each witness sending $176\cdot\numreqsnap$~bytes (112~bytes for each checkpoint and 64~bytes for the witness' signature on it).  This suggests a tradeoff between runtime and bandwidth in settings with low uptimes, with a lower number of witnesses requiring less bandwidth (as we see in Table~\ref{tab:latency}) but a potentially higher runtime if clients have to repeat their requests due to a lack of consensus on the first attempt (as we saw in Section~\ref{sec:gossip-liveness}).

\section{\mogname: A Verifiable Registry}\label{sec:registry}

\subsection{Building blocks: verifiable logs and maps}
\label{sec:maps}

Merkle trees can be used to instantiate a verifiable log, as described in Section~\ref{sec:gossip-defs}.  To achieve the append-only property, entries in the log are ordered chronologically.  This makes it inefficient to use the log as a verifiable \emph{map}, or key-value store, since performing a lookup would require a linear search in order to find the leaf corresponding to the key.  
To enable efficient lookups, entries in a map are typically instead ordered lexicographically.  This makes it difficult to efficiently prove that it is append-only, however, since newly added entries will not necessarily appear as its rightmost leaves.

To construct a verifiable map using a sparse Merkle tree~\cite{revocation-transparency,efficient-smt}, the tree is initialized so that all possible keys in the map are set to some default (null) value $\varepsilon$, and the leaves associated with keys are updated as they take on (real) value.  This means that the only operations in the map are updating its leaves and proving inclusion, which proves the value of a key lookup (against a public map root).  
If keys are hashed using SHA-256 before performing a lookup, then the map can support arbitrary keys and has $2^{256}$ leaves. \al{This whips very quickly over how to use a SMT as a Key$\leftrightarrow$Value dicts, maybe that's common/expected knowledge? } Performing operations on a tree of this size would normally be computationally infeasible, but if the set of keys in the map is sparse within the set of all possible keys then most leaves maintain the default value $\varepsilon$.  This in turn means their parents maintain the default value $H(\varepsilon\|\varepsilon)$, \gary{not sure if it's worth noting that this is vulnerable to multi-tree and inter-tree attacks. CONIKS defines a special value for empty branches.} and most internal nodes in the tree maintain some default value associated with their position in the tree (e.g., their height).  Thus, only the non-default values need to be maintained and revealed in inclusion proofs, which means if there are $N$ keys in the map then an inclusion proof consists of $O(\log N)$ hashes rather than $256$.

\subsection{Our construction: \mogname}

Intuitively, \mogname combines verifiable logs and maps to form a ``map of logs'', meaning a map from keys to append-only logs containing the values that have been associated with this key over time.  To keep track of the versions of the map, its evolving roots are appended to a \emph{map root log} (MRL).  
As discussed above, the server can efficiently prove that the two logs are append-only and can efficiently prove the validity of lookups in the map.  We prove in Theorem~\ref{thm:oscillation} that this combination ensures the security of the overall registry.

More formally, \mogname is comprised of three data structures: a table $\mtable$, a verifiable map $\mmap$, and a log $\mrl$.  The table is a simple lookup table, where (arbitrary) keys map to a list of values.  The keys in the map are the same as in the table, but the values are logs, which we refer to as \emph{leaf} logs and denote using $\leaflog$.  These leaf logs contain all historical values associated with the key, with the current value in the rightmost leaf.  The log $\mrl$ contains map roots in chronological order; i.e., the $i$-th leaf is the root $\roothash_{\mmap,i}$ representing the $i$-th version of the map.  The hash contained in a checkpoint is the root hash of the MRL, which we denote by $\roothash_\mrl$.  A history $\hist$ can be concisely represented by the root and size of the leaf log containing its entries, and the checkpoint at which it was obtained, meaning $\histrepr = (\snap,\roothash_\leaflog,\ell_\leaflog)$. 
Formal specifications of \mogname's algorithms are in Figure~\ref{fig:algs} in Appendix~\ref{sec:formal-registry-algs}.  We provide informal descriptions below.

\begin{sarahlist}

\item[$\lookup$] looks up the history $\hist$ associated with a key $\key$, and in particular the most recent value $\val$ in $\hist$ (we abuse notation slightly and refer to this as $\max(\hist)$).  It must then prove that $\val$ is properly stored, which means providing a path from $H(\val)$ all the way up to the MRL root $\roothash_\mrl$.  This is done in three parts: (1) an inclusion proof $\pi_\leaflog$ of $H(\val)$ in the leaf log, (2) an inclusion proof $\pi_\mmap$ of the leaf log root $\roothash_\leaflog$ in the map, and (3) an inclusion proof $\pi_\mrl$ of the map root $\maproot$ in the MRL.

\item[$\verlookup$] verifies the inclusion proof output by $\lookup$, forming its own hash $H(\val)$ in doing so.  This involves first re-computing the roots of the leaf log and map using $\pi_\leaflog$ and $\pi_\mmap$, as described in Section~\ref{sec:compact} (and in particular Algorithm~\ref{alg:ver-incl}), and then checking inclusion of the computed map root against the (known) MRL root using $\pi_\mrl$.  If the compact ranges used in the leaf log and MRL are only to the left of the node then the proofs also show that (1) $H(\val)$ is the rightmost leaf in $\leaflog$, meaning $\val$ is the latest value, and (2) $\maproot$ is the rightmost leaf in $\mrl$, meaning it represents the latest map version.

\item[$\histlookup$] is similar to $\lookup$ in providing a path to the root $\roothash_\mrl$, but also provides the entries in the leaf log that have been added since the last time the client looked (as represented by $\histrepr_\old$).  This allows the client to re-compute the leaf log root directly, meaning the proof does not need to include $\pi_\leaflog$.

\item[$\verhistlookup$] re-computes the roots of the leaf log and map, using the new entries and $\histrepr_\old$ to compute the former.  It then verifies inclusion in the same way as in $\verlookup$.

\item[$\audit$] requires proving that the stored history was consistent in every intermediate version of the registry since the last time the client accessed its history.  To do this, it first runs $\histlookup$ for every version at which a new value $\val_j$ was appended to the key's history, which results in a proof $\pi_{\hist,j}$.  This proves consistency at those versions, but it also needs to prove that the leaf log didn't change in all the other intermediate versions.  For these, it must (1) prove inclusion of the same leaf log with respect to their (differing) map roots, and (2) prove that these intermediate map roots are the only values that were appended to the MRL.  We can capture this in the following relation $R$.

\begin{align*}
x = &(\roothash_\leaflog, \snap, n), w = (\pi_\mrl,\{\roothash_{\mmap,i},\pi_{\mmap,i}\}_{i=1}^n) \\
&~|~\verincl(\roothash_{\mmap,i},\roothash_{\leaflog},\pi_{\mmap,i})~\forall~i\in[n]~\land \\
&~~~\verincl(\snap,\{\roothash_{\mmap,i}\}_i,\pi_\mrl)
\end{align*}

The server then proves this relation for each changed $\roothash_{\leaflog,j}$, in a way we describe in Section~\ref{sec:relation}, which yields a proof $\pi_{\mathsf{btwn},j}$.  It then returns all of the added values $\val_j$ and their corresponding proofs $\pi_{\hist,j}$ and $\pi_{\btwn,j}$.

\item[$\veraudit$] checks the history at each version at which a value $\val_j$ was added using the proof $\pi_{\hist,j}$, and then checks that the leaf log wasn't changed in the intermediate versions following this (until the next value was added, or until the current version of the registry) using the proof $\pi_{\mathsf{btwn},j}$.  It also uses the values $n_j$ to check that the proofs represent the right number of intermediate map roots, in terms of the different sizes of the MRL.

\gary{Would $\veraudit$ be clearer if phrased in terms of applying $\verhistlookup$ to every version of the map?, or am I misunderstanding what $\veraudit$ is doing}

\end{sarahlist}

\newcommand\registryalgs{
\begin{figure*}
\centering
\begin{tabular}{ll}
\begin{tabular}[t]{l}
\underline{$\lookup(\dict,\key)$}\\
$\hist\gets\dict[\mtable][\key]$\\
$\val\gets\max(\hist)$ \\
$\leaflog\gets\dict[\mmap][\key]$\\
$h_\val\gets \max(\leaflog)$\\
$\pi_\leaflog\gets\proveincl(\leaflog,h_\val)$\\
$\pi_\mmap\gets\proveincl(\mmap,\roothash_\leaflog)$ \\
$\pi_\mrl\gets\proveincl(\mrl,\maproot)$\\
return $\val,\pi=(\pi_\leaflog,\pi_\mmap,\pi_\mrl)$\\
~\\
\end{tabular}
&
\begin{tabular}[t]{l}
\underline{$\verlookup(\snap,\key,\val,\pi)$}\\
$(\pi_\leaflog,\pi_\mmap,\pi_\mrl)\gets\pi$\\
$\roothash_\mrl\gets\snap[\roothash]$ \\
$\roothash_\leaflog\gets\rangerootalg(\mergealg(\pi_\leaflog,H(\val)))$ \\
$(\pi_{\mmap,\ell},\pi_{\mmap,r})\gets\pi_\mmap$ \\
$\roothash_\mmap\gets\rangerootalg(\mergealg(\pi_{\mmap,\ell},\roothash_\leaflog,\pi_{\mmap,r}))$ \\
return $\verincl(\roothash_\mrl,\maproot,\pi_\mrl)\land (|\pi_\mrl| = 1)$
\end{tabular}
~\\
\begin{tabular}[t]{l}
\underline{$\histlookup(\dict,\key,(\snap,\roothash_{\leaflog,\old},\ell_{\leaflog,\old}))$}\\
$\hist\gets\dict[\mtable][\key]$ \\
$(\val,\pi_\leaflog,\pi_\mmap,\pi_\mrl)\gets\lookup(\dict,\key)$\\
$\entries\gets\dict[\leaflog_\key](\ell_{\leaflog,\old},\ell_\leaflog)$\\
return $\entries,\pi=(\pi_\mmap,\pi_\mrl)$
\end{tabular}
&
\begin{tabular}[t]{l}
\underline{$\verhistlookup(\snap,\key,\histrepr_\old,\hist_\new,\pi)$}\\
$(\pi_\mmap,\pi_\mrl)\gets\pi$\\
$(\snap,\roothash_{\leaflog,\old},\ell_{\leaflog,\old})\gets \histrepr_\old$ \\
$\roothash_\leaflog\gets \append(\roothash_{\leaflog,\old},\entries)$ \\
$(\pi_{\mmap,\ell},\pi_{\mmap,r})\gets\pi_\mmap$ \\
$\roothash_\mmap\gets\rangerootalg(\mergealg(\pi_{\mmap,\ell},\roothash_\leaflog,\pi_{\mmap,r}))$ \\
return $\verincl(\roothash_\mrl,\maproot,\pi_\mrl) \land (|\pi_\mrl| = 1)$\\
~\\
\end{tabular}
~\\
\begin{tabular}[t]{l}
\underline{$\audit(\dict,\key,(\snap,\roothash_{\leaflog,\old},\ell_{\leaflog,\old}))$}\\
$\{\val_j\}_{j=1}^{N_\key}\gets\dict[\leaflog_\key](\ell_{\leaflog,\old},\ell_\leaflog)$\\
$\roothash_{\leaflog,0}\gets\roothash_{\leaflog,\old}$ \\
for all $j\in[N_{\key}]$: \\
\quad $h_{\leaflog,j}\gets\append(\val_j,\roothash_{\leaflog,j-1})$ \\
\quad $\pi_{\hist,j}\gets\histlookup(\dict_j,\key,(h_{\leaflog,j-1},\ell_{\leaflog,j-1}))$ \\
\quad $n_j\gets i_j - i_{j-1} + 2$ \\
\quad $x_j\gets (\snap_j,\roothash_{\leaflog,j},n_j)$ \\
\quad $\{\roothash_{\mmap,i}\}_{i=i_{j-1}+1}^{i_j-1} \gets \dict[\mrl](\ell_{j-1},\ell_j)$ \\
\quad $\pi_{\mrl}\gets\proveincl(\mrl,\{\roothash_{\mmap,i}\}_i)$ \\
\quad for all $i\in[i_{j-1}+1, i_j-1]$: \\
\quad\quad $\pi_i\gets\proveincl(\mmap_i, \roothash_{\leaflog,j})$ \\
\quad $w_j\gets (\{\roothash_{\mmap,i}, \pi_i\}, \pi_\mrl)$ \\
\quad $\pi_{\mathsf{btwn},j}\gets\prove(R,x_j,w_j)$ \\
return $\{\val_j, \pi_{\hist,j}, \pi_{\mathsf{btwn},j}, \snap_j\}_{j\in[N_\key]}$
\end{tabular}
&
\begin{tabular}[t]{l}
\underline{$\veraudit(\snap,\key,\histrepr_\old,\hist_\new,\pi)$}\\
$\{\val_j, \pi_{\hist,j}, \pi_{\mathsf{btwn},j}, \snap_j\}_j \gets \pi$ \\
$(\snap_\old,\roothash_{\leaflog,\old},\ell_{\leaflog,\old})\gets \histrepr_\old$ \\
$\roothash_{\leaflog,0}\gets\roothash_{\leaflog,\old}$ \\
$\snap_{N_\key+1}\gets \snap$ \\
for all $j\in[N_\key]$: \\
\quad $\roothash_{\leaflog,j}\gets\append(\val_j,\roothash_{\leaflog,j-1})$ \\
\quad $b_{\hist,j}\gets \verhistlookup(\snap_j,\key,\val_j,\pi_{\hist,j})$ \\
\quad $n_j\gets \snap_{j+1}[\ell] - \snap_j[\ell]$ \\
\quad $x_j\gets(\snap_j,\roothash_{\leaflog,j},n_j)$ \\
\quad $b_{\mathsf{btwn},j}\gets\verify(R,x_j,\pi_{\mathsf{btwn},j})$ \\
return $b_{\hist,1}\land b_{\mathsf{btwn},1}\land\cdots\land b_{\mathsf{btwn},N_{\key}}$ \\
\end{tabular}
\end{tabular}
\caption{The algorithms run by the server (on the left-hand side) and the client (on the right-hand side) for different types of lookup queries.}
\label{fig:algs}
\end{figure*}
}

We now argue for the security of \mogname, in terms of its ability to prevent oscillation attacks.  While we do not consider privacy in this paper, we believe that incorporating a verifiable random function into \mogname would be effective in preventing clients from learning information about other entries in the registry, as has been done before~\cite{coniks,seemless}.

\begin{theorem}\label{thm:oscillation}
If (1) the hash function $H$ is collision resistant, (2) the argument system for the relation $R$ satisfies knowledge soundness, and (3) there is a gossip protocol in place for the MRL that resists split-view attacks, then oscillation attacks (Definition~\ref{def:oscillation}) are not possible in \mogname.
\end{theorem}

A proof of this theorem can be found in Appendix~\ref{sec:registry-proof}.  Intuitively, the collision resistance of the hash function implies membership security of the MRL and leaf log (Definition~\ref{def:membership-log}) and lookup security of the map (Definition~\ref{def:lookup}).  It also makes it impossible to prove that something is the rightmost element in a log if it isn't.  We can thus work our way up, starting with the leaf log and arguing that if an adversary can provide valid inclusion proofs for two conflicting logs then by collision resistance the logs must have inconsistent roots.  Similarly, if they can provide valid inclusion proofs of these two inconsistent leaf log roots in the map, then by lookup security the maps must also have different roots.  If the adversary can provide a valid inclusion proof of this ``different'' map root in the MRL,
then by membership security and knowledge soundness the MRL roots must be inconsistent.  If the adversary can get the clients to accept two inconsistent MRL roots, finally, it has successfully carried out a split-view attack.

\subsection{Proving the audit relation}\label{sec:relation}

We consider two approaches for proving the relation $R$ required for auditing: one that involves simply giving out the witness and one that uses SNARKs to prove knowledge of it.

\subsubsection{Providing the witness}\label{sec:compression}
The na{\"i}ve option for proving this relation is to just provide the witness directly; i.e., to have each proof $\pi_{\btwn,j}$ just be the witness $w_j$.  Verification can then consist of checking each intermediate proof (i.e., running $\verincl(\roothash_{\mmap,i},\roothash_{\leaflog,j},\pi_{\mmap,i})$) and checking the range inclusion proof $\pi_{\mrl,j}$.  
This satisfies knowledge soundness (the extractor just outputs the proof), assuming the membership/lookup security of the MRL and the map, but is clearly inefficient: the combined size of the witnesses is $O(N_\mrl\cdot \log(N))$.  This cost, however, can be reduced.

First, we observe that if the verifier sees every intermediate map root then there is no need to send the inclusion proof $\pi_\mrl$, as they can just incorporate the map roots into the hash in $\snap_{j-1}$ to ensure they get the hash in $\snap_j$.  Second, each map inclusion proof contains $O(\log N)$ hashes, but it is highly unlikely that each of these hashes will change between two adjacent versions.  
We can thus compress the set $\{\pi_i\}_i$ by giving only the hashes that have changed.  
To see how many hashes this is, consider that there are $M$ updates to the map between versions.  For each update $m$, we can define a random variable $X_m$ denoting the depth of the node where the paths to the root from $\roothash_{\leaflog,j}$ and the updated key $m$ intersect.  Because map keys are cryptographic hashes, the variables $X_m$ are independent and identically distributed.  We consider two additional random variables: $C_M = |\{X_m\}_{m=1}^M|$, representing the number of distinct depths at which the $M$ paths intersect the path from $\roothash_{\leaflog,j}$, and $D_M = \max_{m\in[M]}\{X_m\}$, representing the deepest modified hash of the inclusion proof of $\roothash_{\leaflog,j}$.  It is clear that $C_M \leq D_M$, and that $E[C_M]$ is the expected number of hashes that would change.

Each variable $X_m$ follows the geometric distribution with parameter $p = 1/2$.  Intuitively, this is because at each depth we have probability $1/2$ of ``exiting'' the path down from the root, and $X_m$ corresponds to the number of Bernoulli trials until this ``exit'' occurs.  To figure out the expected maximum depth after $M$ updates, we observe that $\frac{1}{\lambda}\sum_{k=1}^M \frac{1}{k} \leq E[D_M] \leq 1 + \frac{1}{\lambda}\sum_{k=1}^M \frac{1}{k}$, where $\lambda = -\ln(1-p)$~\cite{eisenberg}.  In our case $p=1/2$ so $\lambda$ is $-\ln(1/2)=\ln(2)$.  We can also use a formula for the $M$-th harmonic number $H_M = \sum_{k=1}^M \frac{1}{k}$, which states that $H_M = \ln(M) + \gamma + o(1)$, where $\gamma$ is the Euler-Mascheroni constant.  Putting this together, we get that
\begin{align*}
E[D_M] &= O(1) + \frac{1}{\ln(2)}H_M \\
&= O(1) + \frac{1}{\ln(2)}(\ln(M) + \gamma + o(1)) \\
&= O(1) + \log_2(M).
\end{align*}

\begin{figure}
\centering
\includegraphics[width=0.7\linewidth]{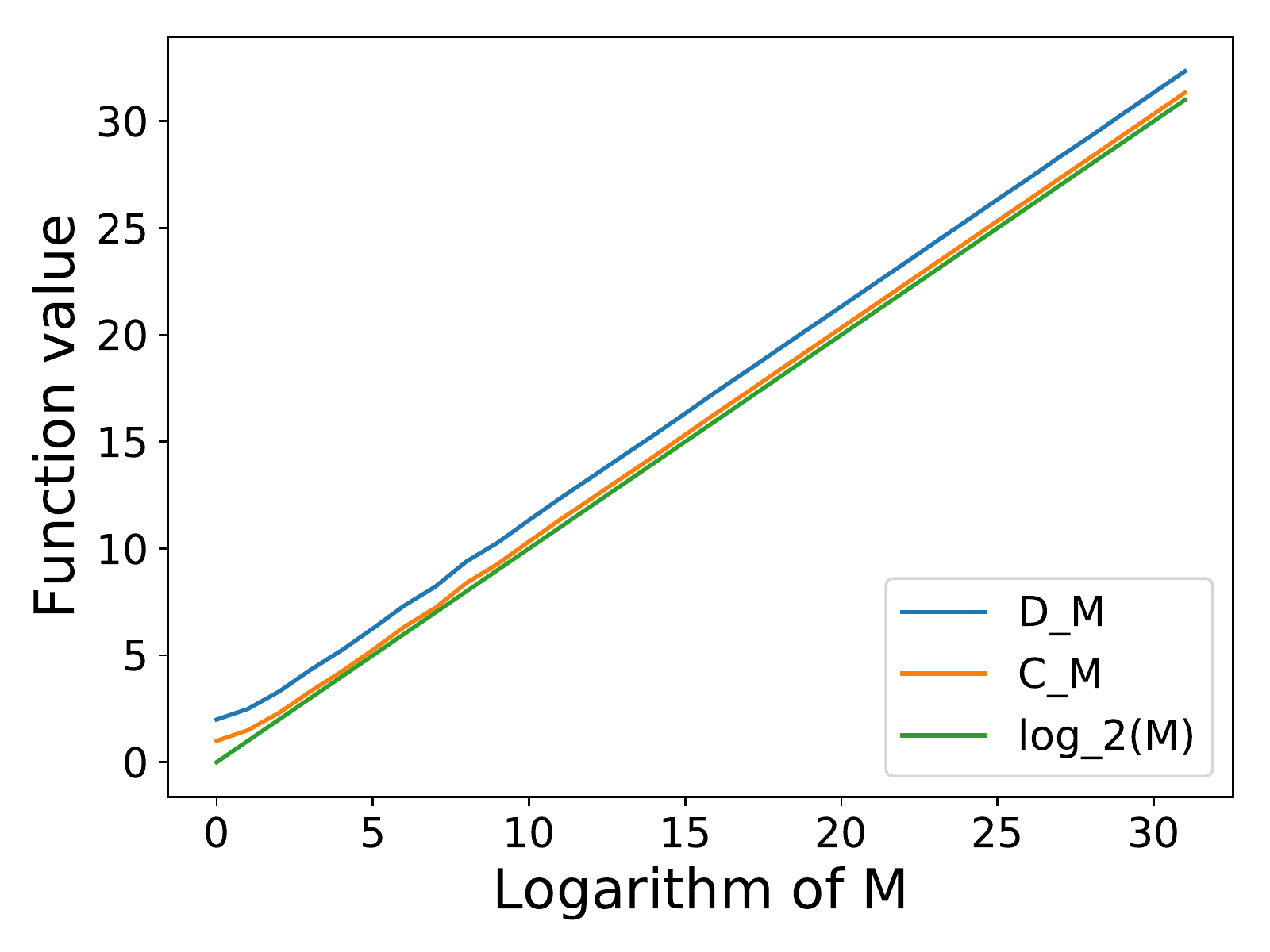}
\caption{The functions $\log_2(M)$, $E[C_M]$, and $E[D_M]$.}
\label{fig:overlap}
\end{figure}

As we can see in Figure~\ref{fig:overlap}, this $O(1)$ term is not hiding a large constant: $E[C_M]$ closely approximates $\log_2(M)$ (within a factor of $1.028$), and $E[D_M] = E[C_M]+1$.  This means that each proof $\pi_i$ needs to contain only $O(\log(M))$ hashes, rather than $O(\log(N))$, so the combined size of the witnesses can be compressed to $O(N_\mrl\cdot \log(M))$.  This is still substantial but we explore its concrete cost in Section~\ref{sec:audit}.

\subsubsection{Using SNARKs}

The more advanced option for proving $R$ would be to use a SNARK (Section~\ref{sec:snark}).  This would yield a constant-sized proof that could be verified in constant time, making the total proof size $O(N_\key)$.\footnote{Getting this bound requires also using a SNARK to prove knowledge of each of the proofs $\pi_{\hist,j}$,
since otherwise these are logarithmic in size.}  
Despite the better asymptotics, we do not currently view the use of SNARKs in this application as practical, as it would incur significant computational overhead for the prover; this is true even for so-called STARKs~\cite{bulletproofs,stark-crypto}, where the size of the proof is also typically logarithmic in the size of the witness.  Moreover, the settings we consider do not have a natural set of parties to fit the model of updatable SNARKs~\cite{updatable-crypto}, so we would need to rely instead on traditional SNARKs~\cite{Groth16}, which would mean introducing a trusted setup.  Nevertheless, we leave it as interesting future research to explore the practicality of SNARKs in this and other Merkle tree-based applications.

\subsection{Implementation and evaluation}
\label{sec:registry-efficiency}

We implemented \mogname in Go, drawing on our own open-source implementation of compact ranges and (sparse) Merkle trees.\footnote{\ifsubmission\else{\url{https://github.com/google/trillian/tree/master/merkle/compact}}\fi}  Our server benchmarks were run in a production environment in which 
each job is run in a worker pool that can be elastically resized.  Our client benchmarks were run on a laptop with an Intel Core i5 2.6\si{\giga\hertz} CPU and 8\si{\giga\byte} of RAM.  It is not possible to compare our performance to previous work throughout, due to the lack of published performance statistics or source code, the different computational environment, and the usage of VRFs (which we do not incorporate) by works like CONIKS~\cite{coniks} and SEEMless~\cite{seemless}, but where possible we do present comparisons with specific measurements.

When evaluating efficiency, we use $N$ to denote the number of entries in the registry (and thus the map), $N_\mrl$ for the number of entries in the MRL, $N_\key$ for the average number of entries in a leaf log, and $M$ for the number of updated entries per version.  %
For $N$, we aim to tolerate up to a billion entries.  For $N_\mrl$ and $N_\key$, we consider the respective rates $r_\mrl$ and $r_\key$ at which new versions and values are created.  
We consider a version of \mogname whose size reflects two continuous years of operation, and pick parameters to fit the three different settings introduced in Section~\ref{sec:gossip-performance} (and defined in full in Appendix~\ref{sec:params}).  
In the aggressive setting, we create a new version once per second and have users update their entries once per hour, so (loosely) bound $N_\mrl$ by $2^{26}$ and $N_\key$ by $2^{15}$.  
In the KT setting, we create a new version once per second and have users update their entries once per month (as they do so only when they replace a device), so $N_\mrl < 2^{26}$ and $N_\key < 2^5$.  
Finally, in the CT setting we create a new version once per day and have users update their entries once per hour, so $N_\mrl < 2^{10}$ and $N_\key < 2^{15}$.  If every key is updated at the rate at which new versions are created then we have that $M = N$.  More generally, we have that $M = N \cdot r_\key / r_\mrl$. 

\subsubsection{Append}\label{sec:append}

Rather than running a loop for all new key/value pairs, we implemented $\append$ as an Apache Beam pipeline.  In particular, we observe that any subtree of a Merkle tree is itself a Merkle tree, so this process can be made parallel by splitting the tree according to height and using the roots of lower-level subtrees (computed in earlier phases of the pipeline) as leaves in upper-level subtrees, until eventually we output the root of the whole tree in the final phase.  \ifsubmission\else{This code is available in the open-source Trillian repository.\footnote{\url{https://github.com/google/trillian/tree/master/experimental/batchmap}}}\fi  
To measure the cost, we take a map of size $N$ and add 100,000 new entries to it (or $N/4$ new entries, whichever is smaller) and update $N/720$ entries.  This corresponds to our KT setting but with slower versioning, as it reflects an average user updating their entries once per month and new versions of the map being created once per hour.%

\begin{figure}
\centering
\includegraphics[width=0.7\linewidth]{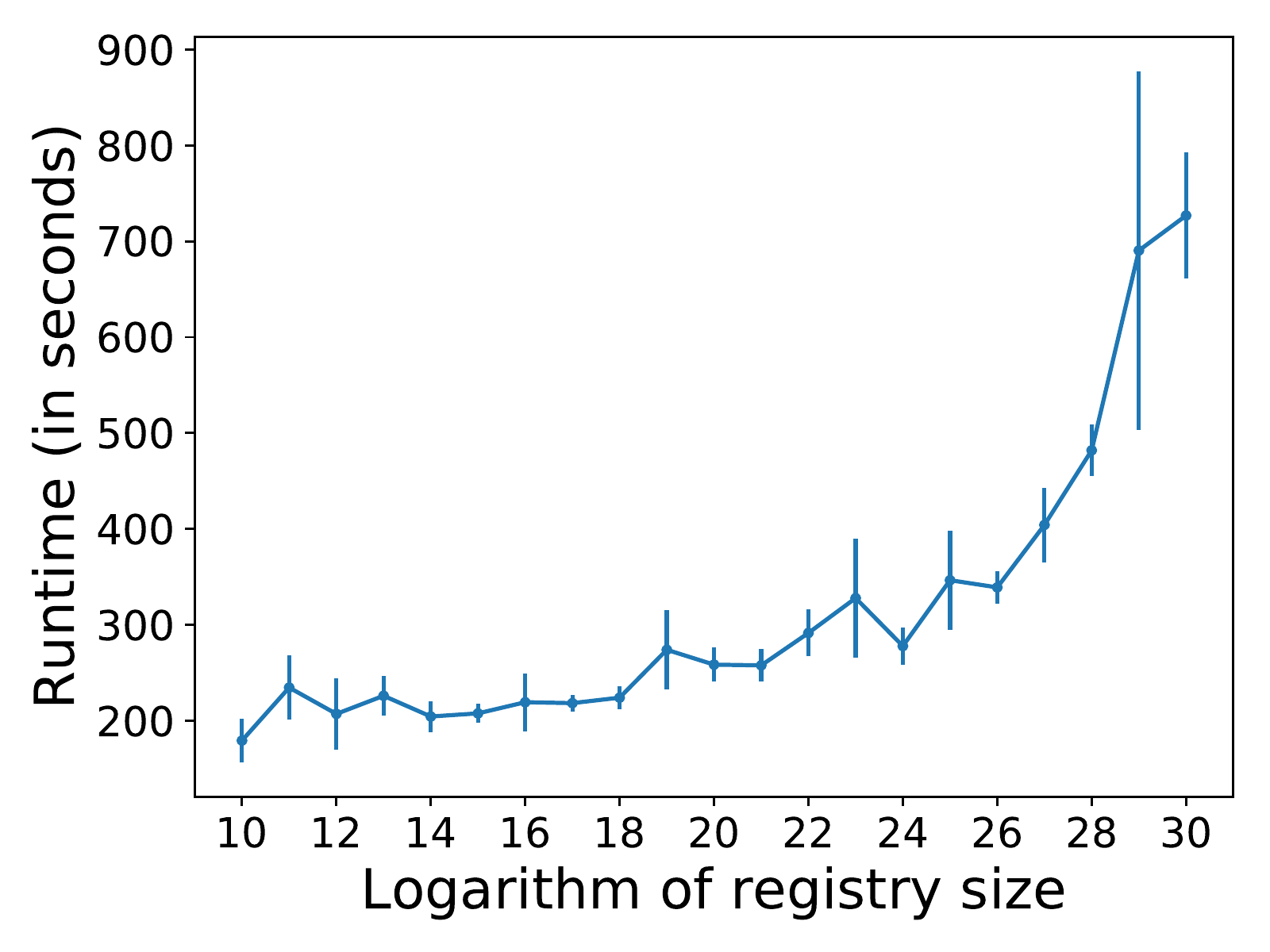}
\caption{The runtime to update the map at the core of \mogname, measured in seconds and averaged over 5 runs.}
\label{fig:append}
\end{figure}

Figure~\ref{fig:append} shows the cost of updating \mogname.  In terms of parallelism, all of our benchmarks for map sizes under $2^{24}$ used under 10 workers; for map sizes under $2^{29}$ they used under 50 workers, and for maps of size $2^{29}$ and $2^{30}$ they used up to 147 workers.  Our pipeline approach imposes more orchestration overhead than one would expect from a purely iterative approach, so it takes longer for smaller map sizes.  On the other hand, it scales to large map sizes, with even a map of $2^{30}$ entries taking under 12 minutes to update on average.  This suggests that this approach is best suited to settings in which versions are produced at, for example, an hourly rate, or in which each version represents a large number of changes.  
As two points of comparison, appending one entry to a tree of size $2^{13}$ took $5$ seconds in an append-only authenticated dictionary (AAD)~\cite{CCS:TBPPTD19}, and inserting 1000 entries into a tree of size 10 million (but updating none) took $2.6$ seconds in CONIKS~\cite{coniks}.

\subsubsection{Lookups}

\begin{figure*}[t]
\centering
\begin{subfigure}[t]{0.32\textwidth}
\includegraphics[width=\linewidth]{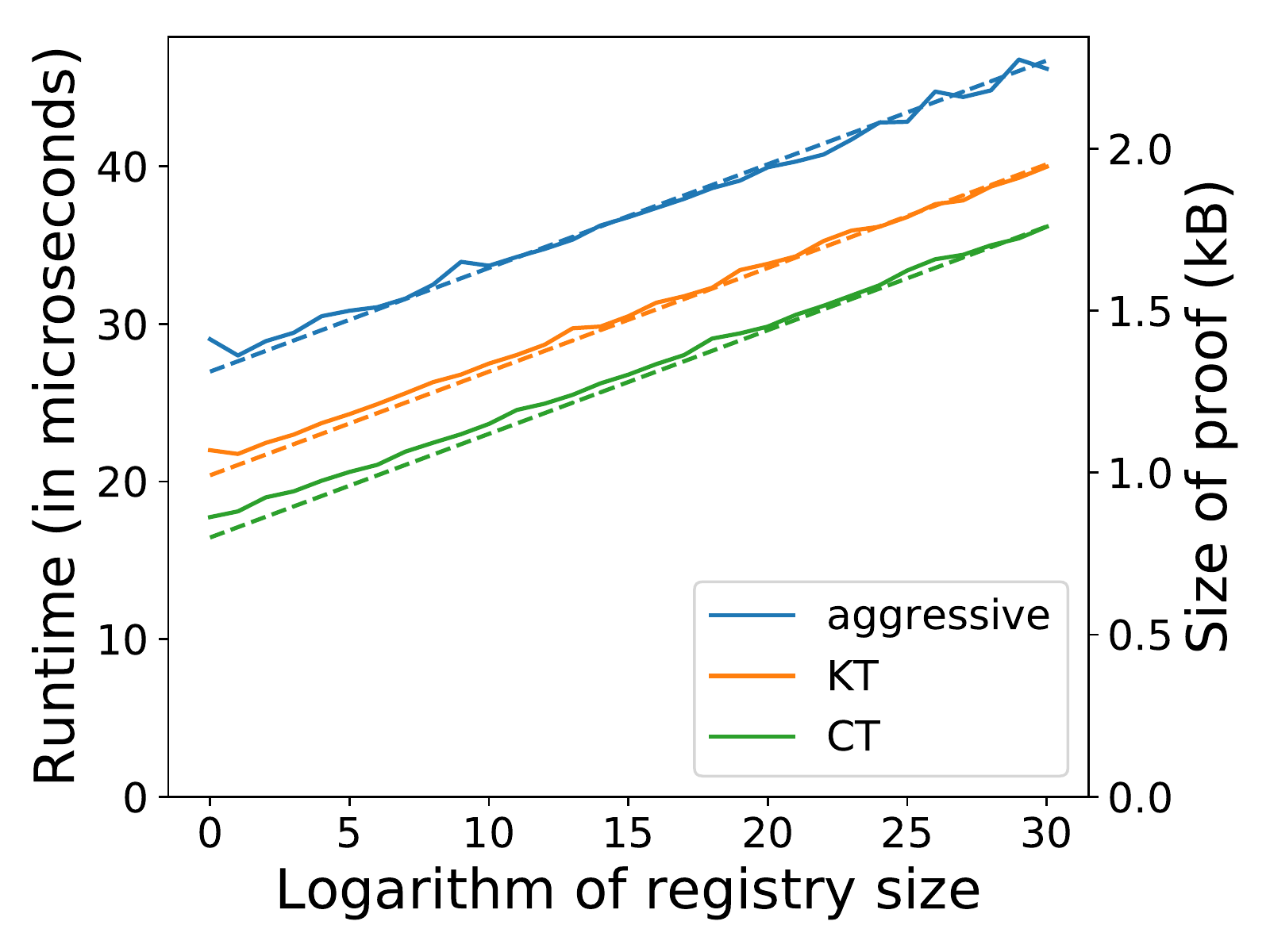}
\caption{$\verlookup$, with runtime in microseconds and proof size in kilobytes.}
\label{fig:lookup}
\end{subfigure}
~
\begin{subfigure}[t]{0.32\textwidth}
\includegraphics[width=\linewidth]{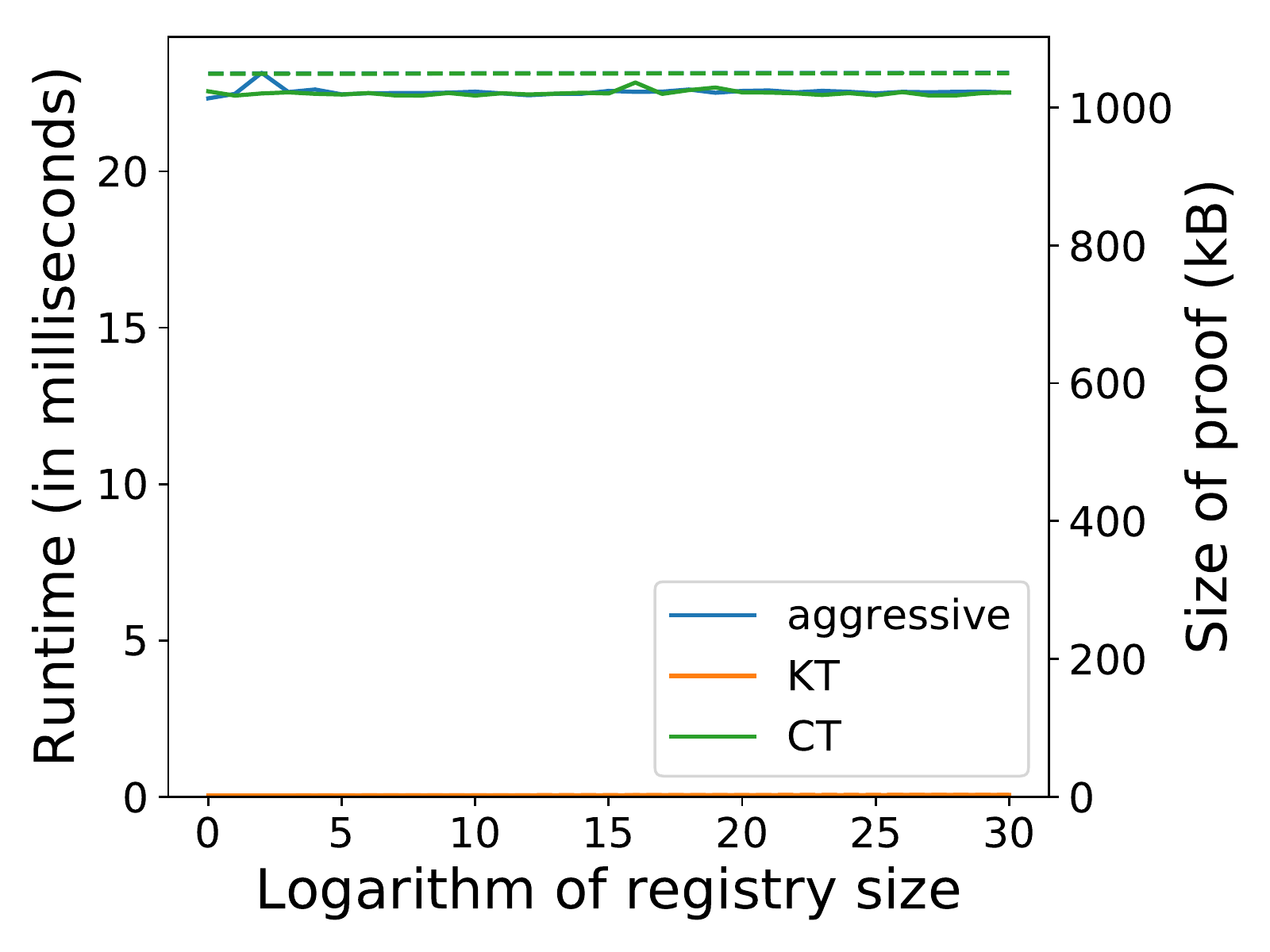}
\caption{$\verhistlookup$, with runtime in milliseconds and proof size in kilobytes.}
\label{fig:hist-lookup}
\end{subfigure}
~
\begin{subfigure}[t]{0.32\textwidth}
\includegraphics[width=\linewidth]{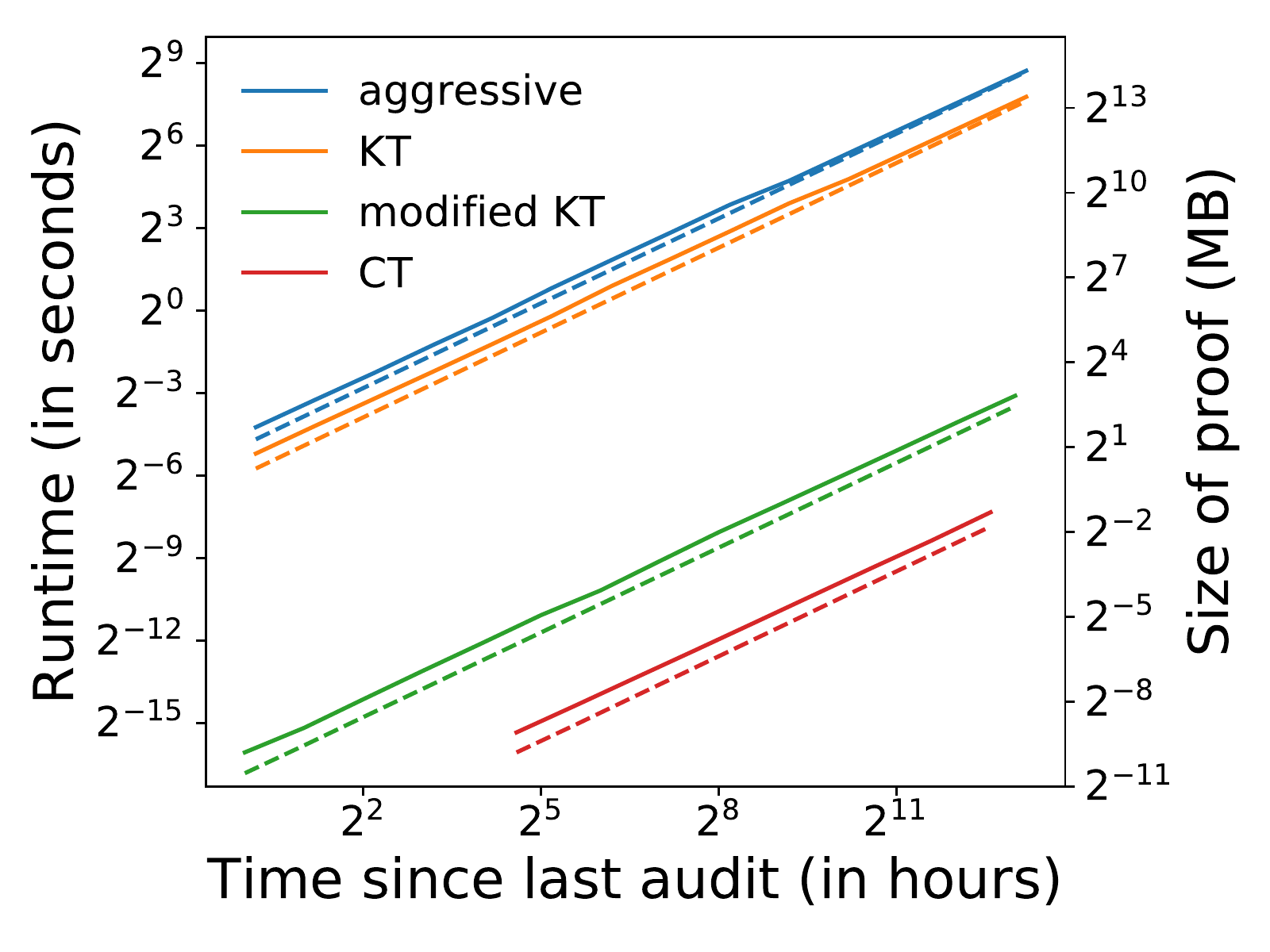}
\caption{$\veraudit$, with runtime in seconds and proof size in megabytes.}
\label{fig:audit}
\end{subfigure}
\caption{The average verification runtime (on the left) and proof size (on the right and using dashed lines) for our three verification algorithms in each of our scenarios, as a function of the registry size (for $\verlookup$ and $\verhistlookup$) and of the time since the last audit (for $\veraudit$).  The average was computed over a minimum of 45 runs for $\verlookup$ and $\verhistlookup$ and a minimum of 2 runs for $\veraudit$.}
\label{fig:lookups}
\end{figure*}

Figure~\ref{fig:lookup} shows the effect that the size of the registry ($N$) has on $\verlookup$, which is logarithmic (as expected).  Even in the aggressive setting, verification took only $46$\si{\micro\second} for a registry of size $2^{30}$, with the proof taking up $2.2$\si{\kilo\byte}.  In SEEMless~\cite[Section 7.2]{seemless}, these costs (excluding the VRF) are reported as over $200$\si{\micro\second} for a registry of 10 million entries.

Figure~\ref{fig:hist-lookup} shows the effect that $N$ has on the runtime of $\verhistlookup$ and the size of proofs.  Unsurprisingly, this cost is dominated not by $\log(N)$ but by the size of the leaf log ($N_\key$), as we benchmark the worst case in which a client starts from scratch; i.e., when $\hist$ contains every value.  In our aggressive and CT settings, in which $N_\key<2^{15}$, proofs still take less than $25$\si{\milli\second} to verify and are slightly more than $1$\si{\mega\byte}.  In the KT setting (barely visible at the bottom), in which $N_\key<2^5$, proofs take at most $58$\si{\micro\second} to verify and are at most $2.8$\si{\kilo\byte}.

\subsubsection{Audits}\label{sec:audit}

The cost of auditing is linear in the number of versions since the last audit, which makes the rate at which versions are produced and the time since the client's last audit the most important factors.
In addition to the regular KT setting, we thus add a ``modified KT'' setting that uses the versioning rate in Section~\ref{sec:append}; i.e., produces a new version every hour instead of every second.

Figure~\ref{fig:audit} shows the runtime of $\veraudit$ and the size of proofs as a function of the time since the most recent audit, ranging from one hour to one year (just over $2^{13}$ hours).  Unsurprisingly, we see much higher costs for the aggressive and KT settings, where versions are produced every second.
Similarly, we see the large effect that lowering the versioning rate has in the KT setting: after waiting a year it takes a client 3.6\si{\minute} to audit a 9.7\si{\giga\byte} proof in the regular setting, as compared to 116.4\si{\milli\second} to audit a 5.5\si{\mega\byte} proof in the modified one.  Given that $\append$ can be performed well within an hour even for large $M$ (as we saw in Section~\ref{sec:append}), this suggests that producing hourly versions is the better tradeoff, despite the fact that it does not benefit as much from the compression discussed in Section~\ref{sec:compression}.  
Finally, we see that bandwidth is the main obstacle to making $\veraudit$ practical, with the highest runtime on the order of minutes but the highest proof size on the order of tens of gigabytes.

\section{Related Work}

\subsection{Gossip protocols}

Chuat et al.\ proposed a gossip protocol for CT in which clients gossip only with web servers, who in turn gossip with auditors~\cite{ct-gossip}.  %
While this makes it harder to disrupt gossip and does not introduce new participants, it introduces privacy concerns in the form of web servers using gossip information to fingerprint and track clients~\cite[Section 10.5.4]{gossiping-in-ct}.  More crucially, it requires a change in a significant fraction of web servers in order to be effective.  Given the speed at which web servers adopt new technology~\cite{USENIX:FBKPBT17,IMC:KRAPVC18}, this makes it unlikely to work in the near future.  

In scenarios like KT that don't have web servers, CONIKS proposed having clients gossip with other clients or with trusted auditors~\cite{coniks}.  It isn't clear how clients would find each other, however, and client-to-client communication isn't scalable.  Some proposals solve this by using a blockchain to disseminate information to clients, which also enables the prevention rather than detection of split-view attacks~\cite{SP:TomDev17,contour,keybase-blockchain}.  This creates an external dependency, however, and furthermore one that is highly inefficient.  
Dahlberg et al.~\cite{dahlberg-aggregation} propose packet aggregation, which also prevents split-view attacks, but requires changing networks themselves.  This is another substantial change that is unlikely to happen in the near future.

More generally, systems such as PeerReview~\cite{peerreview,peerreview-tr} have considered how to reliably detect faulty nodes in distributed systems, based on tamper-evident logging~\cite{secure-timelines}.  These approaches were designed for peer-to-peer networks, rather than networks in which a large set of clients want to gossip about a small set of log operators.  If they were adapted to this setting, our gossip protocol would still offer the benefit of being provably secure, avoiding communication between any parties except clients and witnesses, and enabling the prevention of split-view attacks rather than their retrospective detection.

\subsection{Verifiable data structures}

Perhaps the most closely related verifiable data structure to \mogname is the idea of a verifiable log-derived map (VLDM)~\cite{continusec,usenix-wave}.  In a VLDM, a log contains data associated with a system.  This data can be used to (deterministically) populate a map, whose roots are stored in another log (like our MRL).  A global auditor then re-populates the map using the log data and checks that it gets the same versions of the map/MRL that the server has published.  To some extent, \mogname can be seen as a (provably secure) VLDM with the log sharded into leaf logs, which allows users to perform personal audits (which require $O(N_\mrl\cdot \log(M))$ computation, or even $O(N_\key)$ if using SNARKs, where $N_\mrl$ denotes the number of versions/epochs, $M$ the number of updates per version, and $N_\key$ the number of values per key) 
rather than relying on an entity to audit the entire data structure (which requires $O(N)$ computation, where $N$ denotes the number of entries in the registry and $N > N_\mrl$, $N > M$, and $N \gg N_\key$).  This also enables better detection of oscillation attacks, as oscillations for a single key may be easier to hide within a global log.

In terms of Key Transparency, CONIKS~\cite{coniks} was the first solution that allowed individual users to monitor their keys directly, but they need to rebuild their entire data structure in every version and do not provide a protocol for gossiping about signed tree roots; furthermore, they rely on global auditors to check consistency across different versions of the tree.  This limitation was addressed in EthIKS~\cite{ethiks}, which instead uses the Ethereum blockchain to prevent equivocation.  Similarly, Catena~\cite{SP:TomDev17} proposes storing checkpoints in the Bitcoin blockchain, and Keybase~\cite{keybase-blockchain} stores the root of its Merkle tree in the Stellar blockchain.  Keybase also proposes storing key histories at the leaves of its tree, but does so using a hash chain (which requires a cost of $O(N_\key)$ to verify) as opposed to a log (which requires a cost of $O(\log(N_\key))$).  SEEMless~\cite{seemless} relies on a persistent Patricia trie to maintain history, but needs a global auditor to check that each version is consistent with the previous one and again does not specify how gossip would be instantiated.  Relying on a global auditor means that individual clients do not have to perform audits themselves, in which case Mog and SEEMless have comparable costs ($O(N_\key + \log(N))$) vs.\ $O(N_\key\cdot \log(N))$), but in the absence of a global auditor performing a personal audit in Mog requires $O(N_\mrl\cdot\log(M))$ computation as compared to $O(N_\mrl\cdot M)$ for SEEMless.

More generally, the idea of storing historical values in a dictionary originated with the idea of a persistent authenticated dictionary (PAD)~\cite{isc2001}.  Crosby and Wallach improved on the efficiency of the original design~\cite{ESORICS:CroWal09}, but continued to use the traditional three-party model for authenticated data structures, in which a \emph{source} (the data author) is considered trusted but the \emph{directories} who answer data queries are not.  Pulls and Peeters consider an append-only PAD in the three-party model~\cite{balloon}, in which the source can only append keys but not remove or update old ones.  
Tomescu et al.\ were the first to consider the notion of an append-only authenticated dictionary (AAD) in the two-party model~\cite{CCS:TBPPTD19}, in which a global audit can be performed with only $O(\log N)$ computation.  This is at the cost of increasing the lookup verification time to $O(\log^2 N)$ and the append time to $O(\log^3 N)$, however, and furthermore they do not consider storing historical values and require a trusted setup.  We view it as an interesting open question to see if their techniques can be used to achieve sublinear personal audits, with or without a trusted setup.

\section{Conclusions and Open Problems}

This paper presented a gossip protocol for verifiable logs and a verifiable registry, \mogname, that allows users to perform efficient personal audits rather than relying on a global auditor.  Both of our protocols are provably secure assuming the existence of collision-resistant hash functions, and are performant even under significant loads.  Our experiments suggest that, as the size of registries grows to handle real deployments, it becomes necessary to slow down the rate at which new versions are created.  This further suggests that systems relying on this type of verifiable registry must be able to tolerate retrospective discovery of bad events (e.g., the non-inclusion of public keys or certificates) as opposed to requiring their proactive prevention.  We leave an exploration of this tradeoff as interesting future work.

{\footnotesize
\bibliographystyle{abbrv}

\begin{thebibliography}{10}

\bibitem{minimal-gossip}
\url{https://github.com/google/trillian-examples/tree/master/gossip/minimal}.

\bibitem{contour}
M.~Al-Bassam and S.~Meiklejohn.
\newblock Contour: A practical system for binary transparency.
\newblock In {\em Proceedings of the 2nd International Workshop on
  Cryptocurrencies and Blockchain Technology (CBT)}, 2018.

\bibitem{isc2001}
A.~Anagnostopoulos, M.~T. Goodrich, and R.~Tamassia.
\newblock Persistent authenticated dictionaries and their applications.
\newblock In {\em Proceedings of the 4th International Conference on
  Information Security}, pages 379--393, 2001.

\bibitem{usenix-wave}
M.~P. Andersen, S.~Kumar, M.~AbdelBaky, G.~Fierro, J.~Kolb, H.-S. Kim, D.~E.
  Culler, and R.~A. Popa.
\newblock Wave: A decentralized authorization framework with transitive
  delegation.
\newblock In {\em Proceedings of the 28th USENIX Security Symposium}, 2019.

\bibitem{arpki}
D.~Basin, C.~Cremers, T.~H.-J. Kim, A.~Perrig, R.~Sasse, and P.~Szalachowski.
\newblock {ARPKI: Attack Resilient Public-Key Infrastructure}.
\newblock In {\em Proceedings of ACM CCS 2014}, pages 382--393, 2014.

\bibitem{EC:BelRog06}
M.~Bellare and P.~Rogaway.
\newblock The security of triple encryption and a framework for code-based
  game-playing proofs.
\newblock In {\em Proceedings of Eurocrypt}, 2006.

\bibitem{stark-crypto}
E.~Ben-Sasson, I.~Bentov, Y.~Horesh, and M.~Riabzev.
\newblock Scalable zero knowledge with no trusted setup.
\newblock In {\em Proceedings of Crypto}, pages 701--732, 2019.

\bibitem{ethiks}
J.~Bonneau.
\newblock {EthIKS}: Using {Ethereum} to audit a {CONIKS} key transparency log.
\newblock In {\em Proceedings of the 3rd Workshop on Bitcoin and Blockchain
  Research}, page 95–105.

\bibitem{bulletproofs}
B.~B{\"u}nz, J.~Bootle, D.~Boneh, A.~Poelstra, and G.~Maxwell.
\newblock Bulletproofs: short proofs for confidential transactions and more.
\newblock In {\em Proceedings of the IEEE Symposium on Security \& Privacy},
  2018.

\bibitem{flyclient}
B.~B{\"u}nz, L.~Kiffer, L.~Luu, and M.~Zamani.
\newblock {FlyClient}: super-light clients for cryptocurrencies.
\newblock In {\em Proceedings of the IEEE Symposium on Security \& Privacy},
  2020.

\bibitem{ckps01}
C.~Cachin, K.~Kursawe, F.~Petzold, and V.~Shoup.
\newblock Secure and efficient asynchronous broadcast protocols.
\newblock In {\em Proceedings of Crypto}, pages 524--541, 2001.

\bibitem{seemless}
M.~Chase, A.~Deshpande, E.~Ghosh, and H.~Malvai.
\newblock {SEEMless}: secure end-to-end encrypted messaging with less trust.
\newblock In {\em Proceedings of ACM CCS}, 2019.

\bibitem{CCS:ChaMei16}
M.~Chase and S.~Meiklejohn.
\newblock Transparency overlays and applications.
\newblock In {\em Proceedings of ACM CCS}, 2016.

\bibitem{ct-gossip}
L.~Chuat, P.~Szalachowski, A.~Perrig, B.~Laurie, and E.~Messeri.
\newblock Efficient gossip protocols for verifying the consistency of
  certificate logs.
\newblock In {\em Proceedings of the IEEE Conference on Communications and
  Network Security (CNS)}, 2015.

\bibitem{usec:crowal09}
S.~Crosby and D.~Wallach.
\newblock Efficient data structures for tamper-evident logging.
\newblock In {\em Proceedings of the 18th USENIX Security Symposium}, 2009.

\bibitem{ESORICS:CroWal09}
S.~A. Crosby and D.~S. Wallach.
\newblock Super-efficient aggregating history-independent persistent
  authenticated dictionaries.
\newblock In {\em Proceedings of ESORICS 2009}, pages 671--688, 2009.

\bibitem{efficient-smt}
R.~Dahlberg, T.~Pulls, and R.~Peeters.
\newblock Efficient sparse {Merkle} trees: Caching strategies and secure
  (non-)membership proofs, 2016.
\newblock \url{https://eprint.iacr.org/2016/683.pdf}.

\bibitem{dahlberg-aggregation}
R.~Dahlberg, T.~Pulls, J.~Vestin, T.~H{\o}iland{-}J{\o}rgensen, and A.~Kassler.
\newblock Aggregation-based gossip for {Certificate} {Transparency}, 2018.
\newblock \url{https://arxiv.org/pdf/1806.08817.pdf}.

\bibitem{ESORICS:DGHS16}
B.~Dowling, F.~G{\"u}nther, U.~Herath, and D.~Stebila.
\newblock Secure logging schemes and {Certificate} {Transparency}.
\newblock In {\em Proceedings of ESORICS 2016}, 2016.

\bibitem{dln88}
C.~Dwork, N.~Lynch, and L.~Stockmeyer.
\newblock Consensus in the presence of partial synchrony.
\newblock {\em Journal of the ACM}, 35(2):288--323, 1988.

\bibitem{continusec}
A.~Eijdenberg, B.~Laurie, and A.~Cutter.
\newblock Verifiable data structures, 2015.
\newblock
  \url{github.com/google/trillian/blob/master/docs/VerifiableDataStructures.pdf}.

\bibitem{eisenberg}
B.~Eisenberg.
\newblock On the expectation of the maximum of {IID} geometric random
  variables.
\newblock {\em Statistics \& Probability Letters}, 78:135--143, 2008.

\bibitem{CCS:FDPFSS14}
S.~Fahl, S.~Dechand, H.~Perl, F.~Fischer, J.~Smrcek, and M.~Smith.
\newblock Hey, {NSA}: Stay away from my market! future proofing app markets
  against powerful attackers.
\newblock In {\em Proceedings of ACM CCS}, 2014.

\bibitem{USENIX:FBKPBT17}
A.~P. Felt, R.~Barnes, A.~King, C.~Palmer, C.~Bentzel, and P.~Tabriz.
\newblock Measuring {HTTPS} adoption on the web.
\newblock In {\em Proceedings of the 26th USENIX Security Symposium}, 2017.

\bibitem{ct-pam}
O.~Gasser, B.~Hof, M.~Helm, M.~Korczynski, R.~Holz, and G.~Carle.
\newblock In log we trust: Revealing poor security practices with {Certificate}
  {Transparency} logs and internet measurements.
\newblock In {\em Passive and Active Measurement (PAM) 2018}, pages 173--185,
  2018.

\bibitem{Groth2010}
J.~Groth.
\newblock Short pairing-based non-interactive zero-knowledge arguments.
\newblock In {\em Proceedings of Asiacrypt}, pages 321--340, 2010.

\bibitem{Groth16}
J.~Groth.
\newblock On the size of pairing-based non-interactive arguments.
\newblock In {\em Proceedings of Eurocrypt}, pages 305--326, 2016.

\bibitem{updatable-crypto}
J.~Groth, M.~Kohlweiss, M.~Maller, S.~Meiklejohn, and I.~Miers.
\newblock Updatable and universal common reference strings with applications to
  zk-{SNARKs}.
\newblock In {\em Proceedings of Crypto}, pages 698--728, 2018.

\bibitem{peerreview-tr}
A.~Haeberlen, P.~Kouznetsov, and P.~Druschel.
\newblock The case for {Byzantine} fault detection.
\newblock In {\em Proceedings of HotDep}, 2006.

\bibitem{peerreview}
A.~Haeberlen, P.~Kouznetsov, and P.~Druschel.
\newblock {PeerReview}: Practical accountability for distributed systems.
\newblock In {\em Proceedings of SOSP 2007}, 2007.

\bibitem{keybase-blockchain}
Keybase.io.
\newblock Keybase is now writing to the {Stellar} blockchain, 2020.
\newblock
  \url{https://keybase.io/docs/server_security/merkle_root_in_stellar_blockchain}.

\bibitem{aki}
T.~H.-J. Kim, L.-S. Huang, A.~Perrig, C.~Jackson, and V.~Gligor.
\newblock Accountable key infrastructure {(AKI)}: a proposal for a public-key
  validation infrastructure.
\newblock In {\em Proceedings of WWW 2013}, pages 679--690, 2013.

\bibitem{IMC:KRAPVC18}
P.~Kotzias, A.~Razaghpanah, J.~Amann, K.~G. Paterson, N.~Vallina-Rodriguez, and
  J.~Caballero.
\newblock Coming of age: A longitudinal study of {TLS} deployment.
\newblock In {\em Proceedings of IMC 2018}, 2018.

\bibitem{revocation-transparency}
B.~Laurie and E.~Kasper.
\newblock Revocation {Transparency}, 2012.
\newblock \url{https://www.links.org/files/RevocationTransparency.pdf}.

\bibitem{6962}
B.~Laurie, A.~Langley, and E.~Kasper.
\newblock Certificate {Transparency}, 2013.
\newblock \url{https://tools.ietf.org/html/rfc6962}.

\bibitem{6962bis}
B.~Laurie, A.~Langley, E.~Kasper, E.~Messeri, and R.~Stradling.
\newblock Certificate {Transparency} version 2.0, 2019.
\newblock \url{https://tools.ietf.org/html/draft-ietf-trans-rfc6962-bis-34}.

\bibitem{sundr}
J.~Li, M.~Krohn, D.~Mazieres, and D.~Shasha.
\newblock Secure untrusted data repository ({SUNDR}).
\newblock In {\em Proceedings of the 6th Symposium on Operating Systems Design
  and Implementation (OSDI)}, 2004.

\bibitem{secure-timelines}
P.~Maniatis and M.~Baker.
\newblock Secure history preservation through timeline entanglement.
\newblock In {\em Proceedings of the 11th USENIX Security Symposium}, 2002.

\bibitem{coniks}
M.~S. Melara, A.~Blankstein, J.~Bonneau, E.~W. Felten, and M.~J. Freedman.
\newblock {CONIKS}: Bringing key transparency to end users.
\newblock In {\em Proceedings of the 24th USENIX Security Symposium}, 2015.

\bibitem{merkle}
R.~C. Merkle.
\newblock A digital signature based on a conventional encryption function.
\newblock In {\em Proceedings of Crypto}, pages 369--378, 1987.

\bibitem{chainiac}
K.~Nikitin, E.~Kokoris-Kogias, P.~Jovanovic, N.~Gailly, L.~Gasser, I.~Khoffi,
  J.~Cappos, and B.~Ford.
\newblock {CHAINIAC}: Proactive software-update transparency via collectively
  signed skipchains and verified builds.
\newblock In {\em Proceedings of the 26th USENIX Security Symposium}, 2017.

\bibitem{gossiping-in-ct}
L.~Nordberg, D.~Gillmor, and T.~Ritter.
\newblock Gossiping in {CT}, 2018.
\newblock \url{https://tools.ietf.org/html/draft-ietf-trans-gossip-05}.

\bibitem{sleepy}
R.~Pass and E.~Shi.
\newblock The sleepy model of consensus.
\newblock In {\em Proceedings of Asiacrypt}, pages 380--409, 2017.

\bibitem{balloon}
T.~Pulls and R.~Peeters.
\newblock Balloon: a forward-secure append-only persistent authenticated data
  structure.
\newblock In {\em Proceedings of ESORICS 2015}, pages 622--641, 2015.

\bibitem{enhanced-ct}
M.~D. Ryan.
\newblock Enhanced certificate transparency and end-to-end encrypted mail.
\newblock In {\em Proceedings of NDSS 2014}, 2014.

\bibitem{SP:STVWJG16}
E.~Syta, I.~Tamas, D.~Visher, D.~I. Wolinsky, P.~Jovanovic, L.~Gasser,
  N.~Gailly, I.~Khoffi, and B.~Ford.
\newblock Keeping authorities ``honest or bust'' with decentralized witness
  cosigning.
\newblock In {\em Proceedings of the IEEE Symposium on Security \& Privacy
  (``Oakland'')}, 2016.

\bibitem{unicorn-chi}
J.~Tan, L.~Bauer, J.~Bonneau, L.~F. Cranor, J.~Thomas, and B.~Ur.
\newblock Can unicorns help users compare crypto key fingerprints?
\newblock In {\em Proceedings of ACM CHI}, 2017.

\bibitem{CCS:TBPPTD19}
A.~Tomescu, V.~Bhupatiraju, D.~Papadopoulos, C.~Papamanthou, N.~Triandopoulos,
  and S.~Devadas.
\newblock Transparency logs via append-only authenticated dictionaries.
\newblock In {\em Proceedings of ACM CCS}, 2019.

\bibitem{SP:TomDev17}
A.~Tomescu and S.~Devadas.
\newblock Catena: efficient non-equivocation via {Bitcoin}.
\newblock In {\em Proceedings of the IEEE Symposium on Security \& Privacy
  (``Oakland'')}, 2017.

\end{thebibliography}
\interlinepenalty=10000

}

\appendix

\section{Pseudocode and Proofs for Compact Ranges}
\label{sec:compact-app}

In this section, we provide pseudocode specifications of the compact range algorithms described in Section~\ref{sec:compact}, along with a proof that Algorithm~\ref{alg:range-to-root} produces the correct root.

\compactalg

\compactmergealg

\compactrangerootalg

In Algorithm~\ref{alg:range-to-root}, \textsc{SortedHashStack} returns a stack containing the hashes of the nodes in the compact range in height order; i.e., with the hashes of the nodes at the lowest height at the top of the stack and the hashes of the nodes at the highest height at the bottom.  Since Algorithm~\ref{alg:compact-range} produces nodes in sorted order, and other operations preserve this order, this sorting comes for free.
We now prove that Algorithm~\ref{alg:range-to-root} correctly computes the root of the history tree of size $\ell$.

\lemmarangeroot

\section{Formal Verifiable Log Definitions and a Proof of Theorem~\ref{thm:gossip-safety}}\label{sec:safety-proof}

\formallogdefns

With these definitions in place, we now provide a proof of Theorem~\ref{thm:gossip-safety}, which says that our gossip protocol prevents split-view attacks.

\safetyproof

\section{Parameters Defining Our Settings}\label{sec:params}

We summarize the parameters that comprise the three deployment scenarios we consider in Table~\ref{tab:params}.

\begin{table}
\centering
{\small
\begin{tabular}{lcccc}
\toprule
Parameter & Symbol & Aggressive & KT & CT \\
\midrule
\# servers & $\numsub{S}$ & 1000 & 100 & 100 \\
\# witnesses & $\numwitnesses$ & 25 & 97 & 25 \\
\# clients & $\numsub{C}$ & 1000 & 1000 & 10 \\
Minimum uptime & $U$ & 0.9 & 0.85 & 0.99 \\
Fraction adversarial & $V$ & 4 & 8 &  8 \\
Witness threshold & $Q$ & 16 & 55 & 15 \\
Checkpointing rate & $r_\mrl$ & 1/\si{\second} & 1/\si{\second} & 1/day \\
Rate of user updates & $r_\key$ & 1/\si{\hour} & 1/mo & 1/\si{\hour} \\
Size of MRL & $N_\mrl$ & $2^{26}$ & $2^{26}$ & $2^{10}$\\
Size of registry & $N$  & $2^{30}$ & $2^{30}$ & $2^{30}$\\
Size of leaf log & $N_\key$ & $2^{15}$ & $2^5$ & $2^{15}$ \\
\bottomrule
\end{tabular}
}
\caption{The parameters that define each of the three settings for our evaluations in Section~\ref{sec:gossip-performance} and~\ref{sec:registry-efficiency}.  The fraction of adversarial witnesses is such that $\numwitnesses = VF+1$ and the threshold $Q$ is determined by the formula in Lemma~\ref{lem:bundle-size}.}
\label{tab:params}
\end{table}

\section{Formal Specification of \mogname}
\label{sec:formal-registry-algs}

We formally specify the algorithms that comprise \mogname in Figure~\ref{fig:algs}.

\registryalgs

\section{Formal Verifiable Registry Definitions and a Proof of Theorem~\ref{thm:oscillation}}
\label{sec:registry-proof}

\formalregistrydefns

With these definitions in place, we now provide a proof of Theorem~\ref{thm:oscillation}, which says that \mogname resists oscillation attacks.

\begin{proof}
Let $\A$ be an adversary playing game $\oscillation{\A}{\secp}$.  We use $E_\A$ to denote the event in which $\A$ wins the game, which happens with probability $\pr[\oscillation{\A}{\secp}]$, and build PT adversaries $\{\B_i\}_{i=1}^5$ such that
{\scriptsize
\begin{align*}
\pr[E_\A\land \lnot E_\ks] &\leq \advks{\B_1}{\secp} \\
\pr[E_\A\land E_\ks\land E_\leaflog] &\leq \advcr{\B_2}{\secp} \\
\pr[E_\A\land E_\ks \land \lnot E_\leaflog\land E_\mmap] &\leq \advlookup{\B_3}{\secp} \\
\pr[E_\A\land E_\ks \land \lnot E_\leaflog\land \lnot E_\mmap\land E_\mrl] &\leq \advmemb{\B_4}{\secp} \\
\pr[E_\A\land E_\ks\land \lnot E_\leaflog \land \lnot E_\mmap \land \lnot E_\mrl] &\leq \advsplitview{\B_5}{\secp} 
\end{align*}
}
for independent events $E_\ks$, $E_\leaflog$, $E_\mmap$, and $E_\mrl$ that we define below.  We then have that 
{\scriptsize
\begin{align*}
\advoscillation{\A}{\secp} &= \pr[E_\A] \\ 
&= \pr[E_\A](\pr[\lnot E_\ks] + \pr[E_\ks]) \\
&= \advks{\B_1}{\secp} + \pr[E_\A\land E_\ks] \\
&= \advks{\B_1}{\secp} + \pr[E_\A\land E_\ks](\pr[E_\leaflog] + \pr[\lnot E_\leaflog]) \\
&= \advks{\B_1}{\secp} + \pr[E_\A\land E_\ks \land E_\leaflog]~+ \\
&\quad~~\pr[E_\A \land E_\ks\land \lnot E_\leaflog]\\
&= \advks{\B_1}{\secp} + \advcr{\B_2}{\secp}~+ \\
&\quad~~\pr[E_\A \land E_\ks\land\lnot E_\leaflog](\pr[E_\mmap + \pr[\lnot E_\mmap]) \\
&= \advks{\B_1}{\secp} + \advcr{\B_2}{\secp}~+ \\
&\quad~~\pr[E_\A\land E_\ks\land \lnot E_\leaflog\land E_\mmap]~+ \\
&\quad~~\pr[E_\A\land E_\ks\land \lnot E_\leaflog \land \lnot E_\mmap] \\
&= \advks{\B_1}{\secp} + \advcr{\B_2}{\secp} + \advlookup{\B_3}{\secp}~+ \\
&\quad~~\pr[E_\A\land E_\ks\land \lnot E_\leaflog \land \lnot E_\mmap](\pr[E_\mrl] + \pr[\lnot E_\mrl]) \\
&= \advks{\B_1}{\secp} + \advcr{\B_2}{\secp} + \advlookup{\B_3}{\secp}~+ \\
&\quad~~\pr[E_\A\land E_\ks\land \lnot E_\leaflog \land \lnot E_\mmap \land E_\mrl]~+ \\
&\quad~~\pr[E_\A\land E_\ks\land \lnot E_\leaflog \land \lnot E_\mmap \land \lnot E_\mrl] \\
&= \advks{\B_1}{\secp} + \advcr{\B_2}{\secp} + \advlookup{\B_3}{\secp}~+ \\
&\quad~~\advmemb{\B_4}{\secp} + \advsplitview{\B_5}{\secp},
\end{align*} 
}
from which the theorem follows.  We consider the winning conditions for $\oscillation{\A}{\secp}$; i.e., what happens in event $E_\A$.  If $\A$ succeeds, then it outputs
\[
\snap_1,\snap_2,\key,\val,(\snap_0,\hist_\old),\hist,\pi_1,\pi_2
\]
such that 
\begin{enumerate}
\item $\verlookup(\snap_1,\key,\val,\pi_1) = 1$; 

\item $\veraudit(\snap_2,\key,(\snap_0,\hist_\old),\hist,\pi_2) = 1$;

\item $\version(\snap_0) \leq \version(\snap_1) \leq \version(\snap_2)$; 

\item $\snap_0,\snap_1,\snap_2\in B$;

\item $\val\neq\max(\hist_\old)$; and 

\item $\val\notin\hist$.
\end{enumerate}
We can write $\pi_1 = (\pi_{\leaflog,1},\allowbreak \pi_{\mmap,1}, \pi_{\mrl,1})$ and $\pi_2 = (\pi_{\mmap,2}, \pi_{\mrl,2})$.  We can also define $\roothash_{\leaflog,1}\gets \rangerootalg(\mergealg(\pi_{\leaflog,1},H(\val))$ and $\roothash_{\leaflog,2} \gets \append(\roothash_{\leaflog,0},\hist)$, where $\roothash_{\leaflog,0}$ is the hash of the old leaf log root contained in $\histrepr_\old$.  Similarly, we can define $\roothash_{\mmap,1}$ and $\roothash_{\mmap,2}$ as in the fifth line of $\verlookup$ (i.e., derived from the respective inclusion proofs $\pi_{\mmap,1}$ and $\pi_{\mmap,2}$).

We first consider the event $E_\ks$ in which all of the witnesses $w_j$ are valid, meaning for all $w_j\gets\chi(x_j, \pi_{\btwn,j})$ it is the case that $(x_j,w_j)\in R$.  If $\lnot E_\ks$ happens, meaning there exists a $j$ such that $(x_j,w_j)\notin R$, then we can construct $\B_1$ to break knowledge soundness by outputting $(x_j,\pi_{\btwn,j})$.

From now on, we thus consider that $E_\ks$ does happen.  We next consider the event $E_{\leaflog}$ in which $\roothash_{\leaflog,1}$ is consistent with $\roothash_{\leaflog,2}$, meaning there exists some $j$ such that $\vercommit(\roothash_{\leaflog,1},\entries)=1$ for $\entries\gets\hist_\old\|\hist[1:j]$.  The final two conditions of $E_\A$ tell us, however, that $\val\notin \hist$ and $\val\neq\max(\hist_\old)$, meaning it must be that $\val\neq\max(\entries)$.  By definition of $\mergealg$ and $\rangerootalg$ and the fact that $\verlookup(\snap_1,\key,\val,\pi_1)=1$, there is a node in the leaf log that represents two distinct rightmost descendants: $\val$ and $\max(\entries)$.  This further implies that there are values $x$, $x'$, and $y\neq y'$ such that $H(x\|y) = H(x'\| y')$, which means we can construct $\B_2$ that outputs $x\|y$ and $x'\|y'$ to break collision resistance.

From now on, we assume that $E_\leaflog$ does not happen, meaning there is no prefix $\hist_p$ of $\hist$ such that $\vercommit(\roothash_{\leaflog,1},\hist_\old\|\hist_p)=1$.  We now consider the event $E_\mmap$ in which $\roothash_{\mmap,1}$ is in one of the witnesses $w_j$ extracted from the proofs $\pi_{\btwn,j}$; i.e., there exists a $j$ such that $w_j\gets\chi(x_j,\pi_{\btwn,j})$ and $\roothash_{\mmap,1}\in w_j$.  If $w_j$ is a valid witness ($E_\ks$), we can construct $\B_3$ to break lookup security (of the map) by outputting $(\roothash_{\mmap,1},\key,\roothash_{\leaflog,j},\roothash_{\leaflog,1},\pi_{\mmap,i},\pi_{\mmap,1})$.  This tuple satisfies the first winning condition of the lookup security game (applied to the map) under our assumption that $\roothash_{\leaflog,1} \neq \roothash_{\leaflog,j}$ ($\lnot E_\leaflog$).  By the first condition of $E_\A$, we also know that $\verlookup(\snap_1,\key,\val,\pi_1)=1$ and thus that $\verincl(\roothash_{\mmap,1},\roothash_{\leaflog,1},\pi_{\mmap,1})=1$ (which is equivalent to the second winning condition).  Finally, if $(x_j,w_j)\in R$ then we have values $\roothash_{\leaflog,j}$ (in $x_j$), $\roothash_{\mmap,1}$, and $\pi_{\mmap,i}$ (in $w_j$) such that $\verincl(\roothash_{\mmap,1},\roothash_{\leaflog,j},\pi_{\mmap,i})=1$ (which, given how $\verlookup$ is defined for a map, is equivalent to the third winning condition).  

We finally consider the case that $E_\mmap$ does not happen ($\lnot E_\mmap$), meaning $\roothash_{\mmap,1}$ is not included in any witness $w_j$.  There are again two options: first, $\snap_1$ is consistent with $\snap_2$, meaning there is a list $\entries$, stored with some honest mirror, such that $\vercommit(\snap_2,\entries)=1$ and $\vercommit(\snap_1,\entries[1:k])=1$ for some $k$ ($E_\mrl$).  If $\roothash_{\mmap,1}\notin w_j$ for any $j$ ($\lnot E_\mmap$) and $|S_j| = n_j$ for $x_j = (\snap_j,\roothash_{\leaflog,j},n_j)$ and $w_j = (S_j,\pi_{\mrl,j})$ ($E_\ks$), then it must also be the case that $\roothash_{\mmap,1}\notin \entries$, because $\veraudit$ checks that $\ell_\new = \ell_\old + \sum_{j=1}^{n_\key} n_j$; i.e., that the set of map roots contained in the witnesses represents the entire difference (in terms of the new leaves) between the MRL committed to by $\snap_0$ and the one by $\snap_2$.  We thus have that $\roothash_{\mmap,1}\notin\entries$ and can construct $\B_4$ to break membership security (of the MRL) by outputting $(\snap_1,\roothash_{\mmap,1},\entries[1:k],\pi_{\mrl,1})$.

If instead $\snap_1$ is not consistent with $\snap_2$ ($\lnot E_\mrl$), this essentially means the server has carried out a split-view attack.  We can thus construct $\B_5$ to win at $\splitview{\B_5}{\secp}$ by (1) sending $\snap_0$ and then $\snap_2$ to $\oupdate$ for client $i$, along with a valid proof $\pi\gets\proveappend(\mrl,\snap_0,\snap_2)$; (2) sending $\snap_1$ to $\oupdate$ for client $j$; (3) sending $\roothash_{\mmap,1}$ and $\pi_{\mrl,1}$ to $\ocheck$ for client $j$; and (4) outputting $\entries$.  By the first condition of $E_\A$, $\verincl(\snap_1,\roothash_{\mmap,1},\pi_{\mrl,1})=1$ and thus the first winning condition of $\splitview{\B_5}{\secp}$ is satisfied.  By the third and fourth conditions of $E_\A$, the next two winning conditions are satisfied as well.  Finally, if $\roothash_{\mmap,1}\notin\entries$ ($\lnot E_\mmap$) then the final winning condition is also satisfied.
\end{proof}

\end{document}